\documentclass[11pt,a4paper]{article}

\usepackage[utf8]{inputenc}
\usepackage[T1]{fontenc}
\usepackage[british]{babel}
\usepackage{amsmath,amssymb,amsthm,mathtools}
\usepackage{bm}
\usepackage{braket}                
\usepackage{graphicx}
\usepackage{xcolor}
\usepackage{tikz}
\usetikzlibrary{positioning,arrows.meta,calc,decorations.pathmorphing,fit,backgrounds}
\usepackage{quantikz}
\usepackage{subcaption}
\usepackage{booktabs}
\usepackage{enumitem}
\usepackage[margin=1in]{geometry}
\usepackage[colorlinks=true,linkcolor=blue,citecolor=blue,urlcolor=blue]{hyperref}
\usepackage[capitalise,noabbrev]{cleveref}
\usepackage{algorithm}
\usepackage{algpseudocode}
\usepackage{tikz}
\usetikzlibrary{quantikz2}
\usepackage{subcaption}
\usepackage{soul}
\usepackage[numbers,sort&compress]{natbib}

\theoremstyle{plain}

\newtheorem{proposition}{Proposition}
\newtheorem{corollary}{Corollary}
\theoremstyle{definition}

\newtheorem{example}{Example}
\newtheorem{remark}{Remark}

\usepackage{xspace}

\definecolor{anccol}{HTML}{1F5FA5}    
\definecolor{anccolB}{HTML}{9A4A9B}    
\definecolor{datacol}{HTML}{222222}
\definecolor{seamcol}{HTML}{C8542A}

\title{\textbf{Error-detected surgery on Iceberg codes}}

\author{%
  Andrea Di Fini,$^{1}$\quad
  Samuel Crew,$^{2}$\thanks{ADF and SC contributed equally. Correspondence: \texttt{samuel.crew@qperfect.io}.}\quad
  Laura Pecorari,$^{1}$\quad
  Guido Pupillo$^{1,2,3}$
  \\[1ex]
  \footnotesize $^{1}$University of Strasbourg and CNRS, CESQ and ISIS (UMR 7006), aQCess, 67000 Strasbourg, France\\
  \footnotesize $^{2}$QPerfect, 23 Rue du Loess, 67200 Strasbourg, France\\
  \footnotesize $^{3}$Institut Universitaire de France (IUF), 75000 Paris, France
}

\date{\today}

\begin{document}

\newcommand{\lp}[1]{{\leavevmode\color{orange}[LP : #1]}}

\maketitle

\begin{abstract}
We construct explicit error-detecting surgery gadgets---small systems of auxiliary qubits and checks---for the high-rate Iceberg codes $[[2N,2N-2,2]]$, to perform fault-detected measurements of logical Pauli products. The construction follows the perspective of surgery as the gauging of a logical operator, regarded as a symmetry of the code. We give a complete classification of logical Pauli operators under the permutation automorphism group of the Iceberg code, reducing the construction to one gadget per orbit, and we verify with circuit-level simulations that the gadgets are fault-detecting, with the expected post-selected logical error rate. The gadgets require reconfigurable long-range connectivity, available on platforms such as neutral-atom arrays, making an error-detected demonstration of Pauli-based computation a natural near-term experiment. The paper doubles as a self-contained introduction to gauging and code surgery, developed alongside a simple worked example.
\end{abstract}

\section{Introduction}
\label{sec:intro}

A central challenge in fault-tolerant quantum computation is performing logical gates on encoded data. Pauli-based computation carries out the computation by a sequence of adaptive logical Pauli-product measurements, with a supply of $\ket{T}$ resource states \cite{bravyi2016trading}, rather than applying logical gates directly. The elementary operation here is then the fault-tolerant measurement of a logical Pauli, for which code surgery---building on code deformation and surface-code lattice surgery---is a standard route \cite{bombin2009quantum,horsman2012surface,cowtan2024css}. We review both of these standard ideas with worked examples in \cref{sec:example}.

Recent work of Williamson and Yoder~\cite{williamson2026low} realises surgery in general as the gauging of a logical operator, regarded as a symmetry of the code. Namely, a logical Pauli commutes with every stabiliser check and so leaves the code space invariant, thereby generating a symmetry of the stabiliser Hamiltonian whose ground space is the code space. The mechanism to measure the operator is then code deformation. Rather than acting on the logical operator directly, one adjoins an ancillary system of auxiliary qubits and checks, deforming the code into a larger \emph{merged} code in which the operator is no longer logical but a product of stabilisers. Its eigenvalue is then read off from local check measurements, without disturbing the remaining logical qubits, and finally a reverse deformation, code \emph{split}, returns to the original code. Concretely, the auxiliary qubits are arranged on the edges of a graph over the support of the operator, and the new checks are local Gauss laws on the vertices whose product recovers the operator; this ancillary code system is the surgery \emph{gadget}, and the merge--split procedure described above is the surgery primitive. Other low-overhead approaches with different ancillary system constructions also exist \cite{cohen2022low}, including bridge systems for joint measurements across codes \cite{cross2024improved}. However, for general high-distance codes, these and the auxiliary graph design can be opaque, and near-term experimental demonstrations at scale remain challenging.

In this work, we make a number of contributions. We construct simple explicit error-detecting surgery gadgets for the Iceberg codes $[[2N,2N-2,2]]$ and explain how the merge--split operations may be implemented on hardware with reconfigurable connectivity.\footnote{Neutral-atom platforms are one example of such reconfigurability \cite{bluvstein2022quantum,bluvstein2024logical}, though mid-circuit measurement there is comparatively slow. Nonetheless, surgery is an important route to logic on high-rate codes, where transversal constructions are unavailable in the general case.} We realise the full code surgery machinery --- gauge, merge, split, together with compilation by code automorphisms --- and we perform careful circuit-level numerical simulations to confirm that the gadgets are fault-detecting with the expected $O(p^2)$ suppression of the post-selected logical error rate as a function of the physical error probability $p$. Surgery via gauging is among the most resource efficient routes to fault-tolerant logic on high-rate quantum low-density-parity-check (LDPC) codes \cite{cohen2022low,cross2024improved,williamson2026low,swaroop2026universal}, and an early error-detected demonstration of its basic primitive is therefore a natural experimental target. Rather than pursuing ever larger codes with correspondingly intricate gadgets, here we provide one of the smallest possible examples, arguing that it admits a near-term demonstration of the surgery primitive on reconfigurable hardware. Finally, the paper doubles as a self-contained introduction to gauging, lattice surgery, Pauli-based computation and code automorphisms.

\paragraph{Organisation.}
The paper is organised as follows. We begin with a warm-up example: \Cref{sec:example} discusses Pauli-based computation and gauges a logical symmetry of $H=-X^{\otimes4}-Z^{\otimes4}$, exhibiting the basic surgery mechanism. \Cref{sec:iceberg} reviews $[[2N,2N-2,2]]$ Iceberg codes more generally; we discuss their logical operators and provide a detailed treatment of the automorphism group $S_{2N}$. In \Cref{sec:gadgets} we construct explicit surgery gadgets and simulate them numerically.


\section{Preliminaries}
\label{sec:example}

In this section we review two key ideas. The first is Pauli-based computation, a way to perform universal quantum computation using a series of Pauli-product measurements on qubits. The second is the basic idea of code surgery as the gauging of a global symmetry, allowing us to measure such a logically encoded Pauli string. For the latter, we illustrate the mechanism on a simple spin system that will turn out to be the smallest Iceberg code example (\cref{sec:iceberg}).

\subsection{Pauli-based computation}
\label{sec:pbc-step}

A universal fault-tolerant computation is typically performed by applying a Clifford gate set together with a non-Clifford resource, most commonly the $T$ gate. Pauli-based computation (PBC) reorganises the computation into a form that is well suited to data encoded in a general LDPC code. Rather than applying logical gates, one replaces the entire computation by a (possibly adaptive) sequence of logical Pauli-product measurements, supplemented by a supply of $\ket{T}=(\ket{0}+e^{i\pi/4}\ket{1})/\sqrt2$ resource states \cite{bravyi2005universal,bravyi2016trading}. 

The way that this works follows from a simple observation about Clifford circuits. Measuring $Z_q$ on the output of a Clifford $C$ is the same as measuring the conjugated operator $C^\dagger Z_q C$ on the input state, and since $C$ is Clifford this conjugate is again a Pauli operator. A computational basis readout following any Clifford therefore compiles to a Pauli measurement on the input, with the Clifford itself absorbed into the choice of Pauli. Iterating, a Clifford circuit followed by a $Z$-basis readout reduces to a sequence of Pauli-product measurements whose joint outcome distribution reproduces that of the original circuit, up to a Pauli-frame correction that can be tracked efficiently in classical software.

As a concrete example, let us take four qubits and the GHZ preparation circuit
\begin{equation}
  C = \mathrm{CNOT}_{2\to3}\,\mathrm{CNOT}_{1\to2}\,\mathrm{CNOT}_{0\to1}\,\mathrm{H}_0,
\end{equation}
followed by a $Z$-basis readout of all four qubits. Conjugating each readout back through the circuit gives
\begin{equation}
  C^\dagger Z_0 C = X_0,\quad
  C^\dagger Z_1 C = X_0Z_1,\quad
  C^\dagger Z_2 C = X_0Z_1Z_2,\quad
  C^\dagger Z_3 C = X_0Z_1Z_2Z_3,
\end{equation}
so the entire computation collapses to four Pauli measurements with outcomes $(m_0,\dots,m_3)$. The output bits of the original circuit are recovered from these outcomes up to a Pauli-frame correction. The circuits are illustrated in  \cref{fig:pbc-warmup}. We return to this circuit in \cref{sec:gadgets}, where each measurement is compiled to a surgery gadget on a single $[[6,4,2]]$ Iceberg code block as a worked example.

\begin{figure}[t]
  \centering
  \begin{subfigure}[c]{0.42\textwidth}
    \centering
    \begin{quantikz}[row sep=0.35cm, column sep=0.40cm]
      \lstick{$q_0$} & \gate{H} & \ctrl{1} & \qw      & \qw      & \meter{} & \rstick{$z_0$}\setwiretype{c} \\
      \lstick{$q_1$} & \qw      & \targ{}  & \ctrl{1} & \qw      & \meter{} & \rstick{$z_1$}\setwiretype{c} \\
      \lstick{$q_2$} & \qw      & \qw      & \targ{}  & \ctrl{1} & \meter{} & \rstick{$z_2$}\setwiretype{c} \\
      \lstick{$q_3$} & \qw      & \qw      & \qw      & \targ{}  & \meter{} & \rstick{$z_3$}\setwiretype{c}
    \end{quantikz}
    \caption{Gate circuit with $Z$-basis readout.}
    \label{fig:pbc-gates}
  \end{subfigure}
  \hfill
  \raisebox{3.0em}{\large$\;\equiv\;$}
  \hfill
  \begin{subfigure}[c]{0.46\textwidth}
    \centering
    \begin{quantikz}[row sep=0.35cm, column sep=0.55cm]
      \lstick{$q_0$} & \gate{M_X} & \gate[2]{M_{XZ}} & \gate[3]{M_{XZZ}} & \gate[4]{M_{XZZZ}} & \qw \\
      \lstick{$q_1$} & \qw      & \qw          &               &                & \qw \\
      \lstick{$q_2$} & \qw      & \qw          & \qw           &                & \qw \\
      \lstick{$q_3$} & \qw      & \qw          & \qw           & \qw            & \qw
    \end{quantikz}
    \caption{Compiled PBC form.}
    \label{fig:pbc-ppm}
  \end{subfigure}
  \caption{Pauli-based compilation of the four-qubit GHZ example. The Clifford circuit (a) compiles to the sequence of Pauli-product measurements (b), with outcomes $(m_0,\dots,m_3)$ reproducing the readout bits up to a Pauli frame correction.}
  \label{fig:pbc-warmup}
\end{figure}
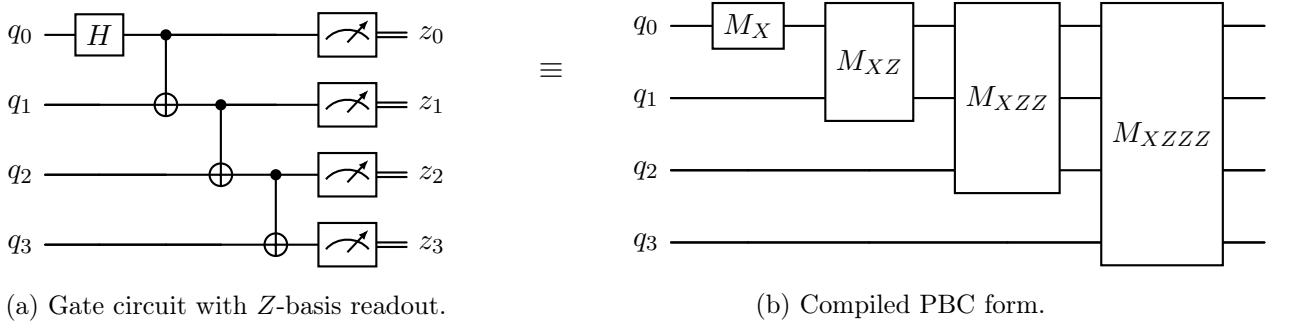

Non-Clifford gates may also be performed in the same framework. A $T$ gate is not applied directly but teleported into the computation by consuming a resource state $\ket{T}=(\ket{0}+e^{i\pi/4}\ket{1})/\sqrt2$: one adjoins the magic state as an extra qubit and performs a joint $Z\!\otimes\!Z$ measurement between it and the data qubit, followed by a Clifford correction $S$ conditioned on the outcome. Each $T$ gate thus enlarges the register by one ancilla, contributes one further Pauli-product measurement, and leaves the residual Clifford carrying a conditional $S$. This is what makes the sequence adaptive since a conditional $S$ is Clifford but not Pauli, so commuting it past a later measurement rotates the affected factors --- for instance $S X S^\dagger = Y$ --- and the Pauli actually measured depends on the earlier outcome. The measurements must therefore be taken in order, each measurement type conditioned on the results so far. 

A fault-tolerant realisation of PBC on encoded data therefore reduces to a fault-tolerant realisation of a Pauli string measurement primitive. The remainder of the paper constructs an error-detecting version of this primitive for data encoded in an Iceberg code. These ideas can be made transparent on a simple spin system, our warm-up example, to which we now turn.

\subsection{Warm-up example of code surgery}
\label{sec:warmup}
Pauli-based computation requires a primitive that non-destructively measures a logical Pauli operator. As a warm-up example, we show how such a measurement is realised by {\it gauging}, working with a simple system of four spins $\mathcal{H} = (\mathbb{C}^2)^{\otimes 4}$ with dynamics described by 
\begin{equation}
    H = -X_0X_1X_2X_3 - Z_0Z_1Z_2Z_3 .
\end{equation}
The ground space is the simultaneous $+1$ eigenspace of the two commuting operators $s_X = X_0X_1X_2X_3$ and $s_Z = Z_0Z_1Z_2Z_3$. This ground space is the code space of the stabiliser code~\cite{gottesman1997stabilizer,poulin2005stabilizer} with stabiliser group $\mathcal{S} = \langle s_X, s_Z\rangle$, encoding two logical qubits.\footnote{This is the smallest member of the Iceberg code family introduced in \cref{sec:iceberg}, the $N=2$ case. We treat it here purely as a spin system, and defer the general construction to \cref{sec:gadgets}.} Its logical operators are
\begin{equation}
    \bar{X}_0=X_0X_1,\quad \bar{Z}_0=Z_1Z_3,\qquad \bar{X}_1=X_0X_2,\quad \bar{Z}_1=Z_2Z_3.
\end{equation}
Let us consider the task of measuring $\bar{X}_0$. The idea is to interpret the operator as a global symmetry of $H$ and then gauge that symmetry so that it becomes a local constraint of an associated gauged theory \cite{williamson2026low}. The important feature of this construction is that the ground space of the gauged Hamiltonian coincides with the code space of a new code, the \emph{merged code}, in which the $\bar{X}_0$ eigenvalue can be reconstructed by measuring local stabilisers.

The operator $\bar{X}_0 = X_0X_1$ commutes with $H$ and so is a symmetry ($\bar{X}_0^\dagger H \bar{X}_0 = H$), acting on operators by conjugation. Since it anticommutes with $Z_0$ and $Z_1$, the symmetry acts as $Z_i \mapsto -Z_i$ for $i = 0, 1$, leaving all other single-site Pauli operators unchanged.

Gauging this symmetry means promoting the global $\mathbb{Z}_2$ generated by $\bar{X}_0$ to an independent action on each of qubits $0$ and $1$, enlarging the symmetry group to $\mathcal{G}=\mathbb{Z}_2\times\mathbb{Z}_2$, with generic element $T(f_0,f_1)=X_0^{f_0}X_1^{f_1}$, $f_0,f_1\in \mathbb{Z}_2$. The Hamiltonian is not invariant under $\mathcal{G}$ since, while $s_X$ commutes with every $T(f_0,f_1)$, the stabiliser $s_Z$ transforms as
\begin{equation}
    T(f_0,f_1)^\dagger\, s_Z\, T(f_0,f_1) = (-1)^{f_0+f_1}\, s_Z,
\end{equation}
differing from $s_Z$ precisely when $f_0+f_1$ is odd. To restore invariance, we introduce a gauge qubit $\tau_{01}$ on a link joining sites $0$ and $1$, extending the Hilbert space to $\mathcal{H}' = (\mathbb{C}^2)^{\otimes 4}_{\text{matter}} \otimes \mathbb{C}^2_{\text{gauge}}$. We denote its Pauli operators $\tau^X$ and $\tau^Z$, placing the Pauli type in a superscript to distinguish gauge operators from those acting on the matter qubits, and we drop the link label since there is only one link.

The operator $\tau^Z$ transforms as $\tau^Z \mapsto (-1)^{f_0+f_1}\tau^Z$ (with $\tau^X$ invariant) so that under a gauge transformation it acquires the product of the $\mathbb{Z}_2$ charges at the two endpoints of its link. This is exactly how a parallel transport variable on a link responds to gauge rotations at its endpoints in a lattice gauge theory \cite{kogut1975hamiltonian,kogut1979introduction}, which is why $\tau^Z$ is referred to as a gauge field. Its transformation ensures that the minimally coupled term
\begin{equation}
    \tilde{s}_Z \equiv s_Z\, \tau^Z
\end{equation}
is gauge invariant. Consistency then fixes the dressed generators of $\mathcal{G}$ on $\mathcal{H}'$ as
\begin{equation}
    G_0=X_0\,\tau^X, \qquad G_1=X_1\,\tau^X,
\end{equation}
where each bare $X_i$ is dressed by $\tau^X$ to account for the degree of freedom of the link.

The minimal coupling prescription then replaces $s_Z$ by $\tilde{s}_Z$, giving the gauged Hamiltonian
\begin{equation}
    \tilde{H} = -X_0X_1X_2X_3 - Z_0Z_1Z_2Z_3\,\tau^Z,
\end{equation}
subject to the Gauss-law constraints $G_i\,\ket{\psi} = \ket{\psi}$ for $i=0$ and $i=1$. Its ground space, restricted to this physical subspace, is the code space of a stabiliser code with group $\langle s_X,\,\tilde{s}_Z,\,G_0,\,G_1\rangle$, the \emph{merged code}. Counting gives five qubits with four independent stabilisers, so a single logical qubit remains, with explicit corresponding Pauli operator representatives
\begin{equation}
    \bar{X}=X_0X_2,\qquad \bar{Z}=Z_2Z_3.
\end{equation}
Gauging has thus consumed the logical qubit associated to $\bar{X}_0=X_0X_1$ and in the merged code $\bar{X}_0$ is no longer logical but a product of stabilisers, $\bar{X}_0=G_0G_1$. Following a suitable initialisation and measurement protocol, reviewed later in \cref{sec:gadgets}, the eigenvalue of $\bar{X}_0$ is recovered by repeatedly measuring the merged-code stabilisers and computing $G_0G_1$. 

This four-spin example already contains the essential gauging mechanism of code surgery. In the remainder of the work we generalise to the Iceberg codes, of which this is the smallest example, and build fault-detecting schemes to measure logical operators. The ancilla system introduced to gauge an operator --- here just the single gauge qubit $\tau$ --- is referred to as a \emph{surgery gadget}, and in \cref{sec:gadgets} we construct gadgets for the general Iceberg code.

\section{The Iceberg code}
\label{sec:iceberg}

The Iceberg code \cite{gottesman1998theory} has a very high encoding rate (the logical to physical qubit ratio) but distance $d=2$, so it can detect but not correct errors. The Iceberg nomenclature, and its role as a low-overhead target for near-term hardware, are due to a trapped-ion demonstration \cite{self2024protecting}, where post-selection on the detection of single faults suppresses the effective logical error rate. Iceberg codes have since also been demonstrated on neutral-atom platforms, with experiments ranging from logical state preparation and circuits to full error-detected applications \cite{reichardt2024logical,zhang2026logical,lib2026velocity,mathiot2026benchmarking}. 

In this section we review the Iceberg code and study its permutation automorphism group. Automorphisms partition the logical Pauli operators of the code into orbits, and, as we will see in more detail in the following, operators in the same orbit can be measured by the same surgery gadget (constructed in \cref{sec:gadgets}) so that the measurement of exponentially many logical Paulis requires only polynomially many gadgets.

\subsection{Definition and logical operators}
\label{sec:iceberg-def}

For $N \ge 2$, the Iceberg code $[[2N,2N-2,2]]$ is the Calderbank–Shor–Steane (CSS) code on $2N$ data qubits defined by the two weight-$2N$ stabiliser generators
\begin{equation}
  s_X = X^{\otimes 2N},\qquad s_Z = Z^{\otimes 2N},
\end{equation}
or equivalently as the ground space of the spin Hamiltonian $H = -X^{\otimes 2N} - Z^{\otimes 2N}$; the warm-up of the previous section is the $N=2$ case. We illustrate a physical layout of the qubits in \Cref{fig:iceberg-complex}.

\begin{figure}[t]
  \centering
  \includegraphics[width=0.5\textwidth]{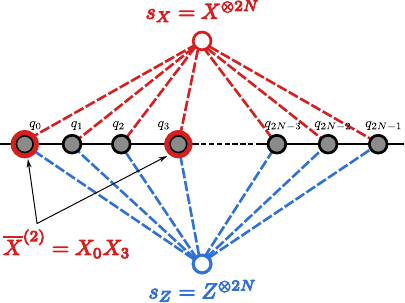}
  \caption{A physical layout of the $[[2N,2N-2,2]]$ Iceberg code qubits.}
  \label{fig:iceberg-complex}
\end{figure}

The code encodes $2N-2$ logical qubits, which we index by $i\in\{0,\dots,2N-3\}$ with logical operator representatives
\begin{equation}\label{eq:logop}
    \bar{X}^{(i)}=X_0 X_{i+1},\qquad \bar{Z}^{(i)}=Z_{i+1}Z_{2N-1}.
\end{equation}
In this representation qubit $0$ is shared by all $\bar{X}^{(i)}$ and qubit $2N-1$ by all $\bar{Z}^{(i)}$, while each middle qubit $1,2,\dots,2N-2$ carries one logical qubit. These are minimal weight representatives, so the code has distance $d=2$. 

As argued by Chao and Reichardt \cite{chao2018fault}, measuring the stabilisers $s_X$ and $s_Z$ fault-tolerantly requires a flag qubit. Bare extraction is not fault tolerant: since every even-weight single-type Pauli string is a logical operator, an ancilla fault propagates as a hook error that is undetectable for any CNOT ordering. A flag qubit sandwiched around the syndrome extraction CNOTs catches these errors and restores circuit-level distance two.

\subsection{Automorphism group}
\label{sec:aut}

The permutation group $S_{2N}$ acts on the physical Hilbert space $\mathcal{H} =(\mathbb{C}^{2})^{\otimes 2N}$ by permuting the tensor factors. On a computational basis state for $\sigma \in S_{2N}$ a permutation acts as
\begin{equation}
   \sigma \ket{i_0 i_1 \cdots i_{2N-1}} = \ket{i_{\sigma^{-1}(0)} i_{\sigma^{-1}(1)} \cdots i_{\sigma^{-1}(2N-1)}}.
\end{equation}
The (permutation) automorphism group of a code is defined as the subgroup of permutations that leave the code space (\textit{i.e.} the ground space in the language of the warm-up example) invariant \cite{sayginel2025fault,berthusen2025automorphism}. In general, the permutation automorphism group acts\footnote{A permutation of qubits maps any Pauli string to a Pauli string of the same type, so conjugation by it sends the stabiliser group to itself and normalises the logical Pauli group, so the induced logical action is therefore Clifford. Moreover, since $X$-type strings map to $X$-type and $Z$-type to $Z$-type, the logical action lies in the subgroup generated by CNOTs.} as a Clifford unitary on the encoded logical qubits.

Now, since the stabilisers of the Iceberg code are $X^{\otimes 2N}$ and $Z^{\otimes 2N}$ (both invariant under $S_{2N}$), its permutation automorphism group coincides with the whole of $S_{2N}$, and so it remains only to determine the induced Clifford logical action of this group on the encoded logical qubits.

To this end, we consider a convenient set of generators for $S_{2N}$. Labelling the physical qubits of the Iceberg code as $\{0,1,\ldots,2N-1\}$, we introduce the terminology \textit{left/right boundary} to denote the qubits $0$ and $2N-1$ and we call $\{1,2,\ldots,2N-2\}$ the \textit{bulk} qubits. It is clear that $S_{2N}$ is generated by bulk transposition $(i\; i+1)$ for $i\in \{1,2,\dots, 2N-3\}$, left boundary transposition $(0\;1)$ and right boundary transposition $(2N-2\;2N-1)$.

We find that bulk transposition generates a logical \textsf{SWAP} between logical qubits $i-1$ and $i$ for each $i\in\{1,\dots,2N-3\}$, left boundary transposition generates a \textsf{FAN-IN} onto logical qubit $0$, and right boundary transposition generates a \textsf{FAN-OUT} from logical qubit $2N-3$; explicitly,
\begin{equation}
\begin{aligned}
  (i\;\;i{+}1) &\;\longmapsto\; \overline{\mathrm{SWAP}}_{i-1,i},
    && i\in\{1,\dots,2N-3\}, \\[2pt]
  (0\;\;1) &\;\longmapsto\; \prod_{k=1}^{2N-3}\overline{\mathrm{CNOT}}_{k\to 0}, \\[2pt]
  (2N-2\;\;2N-1) &\;\longmapsto\; \prod_{k=0}^{2N-4}\overline{\mathrm{CNOT}}_{2N-3\to k }.
\end{aligned}
\end{equation}
where an overline indicates a logical gate.

\subsection{Orbits and seed operators}
\label{sec:orbits}
We now group logical operators into orbits under the action of $S_{2N}$. Two operators belong to the same orbit if and only if they are related by one of the permutation-induced circuits discussed in the previous subsection. Classifying these orbits is useful because constructing a single surgery gadget per orbit suffices to measure every logical operator, as we show in the following. This approach was recently elucidated for general LDPC codes by Webster, Smith and Cohen \cite{webster2025explicit}.

To begin, we note that each logical operator $\bar{P}$ has a physical representative $P$, to which we associate a tuple
\begin{equation}
    t(P)=(n_I,n_X,n_Y,n_Z),
\end{equation}
counting how many of the $2N$ sites carry each Pauli type.\footnote{For example, the logical operator $\bar{P}=\bar{Y}^{(0)}$ has physical representative $P = XYIIIZ$, corresponding to $t(P) = (3,1,1,1)$.} Valid logical operators of the Iceberg code satisfy $\sum_\alpha n_\alpha = 2N$, $n_X+n_Y\equiv n_Z+n_Y\equiv 0\pmod{2}$, and $n_\alpha < 2N$ for every $\alpha \in \{I,X,Y,Z\}$.\footnote{A Pauli string $P$ commutes with $s_X=X^{\otimes 2N}$ if it anticommutes on an even number of sites, giving $n_Z+n_Y\equiv 0\pmod{2}$, and similarly commuting with $s_Z=Z^{\otimes 2N}$ gives $n_X+n_Y\equiv 0\pmod{2}$. The condition $n_\alpha<2N$ for every $\alpha$ ensures $P$ is not itself a stabiliser.} Since a permutation only rearranges letters without changing their counts, it fixes $t(P)$, and $S_{2N}$ acts transitively on all strings sharing a given tuple. The remaining freedom is then multiplication by elements of the stabiliser group $\{I,s_X,s_Z,s_Xs_Z\}$, which acts on tuples by
\begin{equation}\label{eq:stabilizer_group_action}
\begin{aligned}
  s_X&:\,(n_I,n_X,n_Y,n_Z)\mapsto(n_X,n_I,n_Z,n_Y),\\
  s_Z&:\,(n_I,n_X,n_Y,n_Z)\mapsto(n_Z,n_Y,n_X,n_I),\\
  s_Xs_Z&:\,(n_I,n_X,n_Y,n_Z)\mapsto(n_Y,n_Z,n_I,n_X).
\end{aligned}
\end{equation}
This stabiliser group action partitions the set of valid tuples into equivalence classes, or orbits. As we prove in \Cref{app:orbits}, the $S_{2N}$-orbits of logical Pauli strings of the Iceberg code are in one-to-one correspondence with these equivalence classes, and the number of orbits is given by
\begin{equation}
  \#\text{orbits}
  =
  \begin{cases}
    \dfrac{N^3+3N^2+5N-9}{12}, & N \text{ odd},\\[2ex]
    \dfrac{N^3+3N^2+14N}{12}, & N \text{ even}.
  \end{cases}
\end{equation}
Finally, within each orbit we define a \emph{seed operator} as the representative of minimal weight $w = 2N - n_I$ with the form
\begin{equation}
    P_{\text{seed}} = \underbrace{X\cdots X}_{n_X}\,\underbrace{Y\cdots Y}_{n_Y}\,\underbrace{Z\cdots Z}_{n_Z}\,\underbrace{I\cdots I}_{n_I}.
\end{equation}
Every operator in the orbit is then obtained from $P_{\text{seed}}$ by an element of $S_{2N}$.

\Cref{alg:seed} below summarises and formalises this idea: given any logical Pauli, it returns the seed of its orbit together with the permutation $\sigma\in S_{2N}$ carrying the operator to that seed, and we now give an example of this procedure.

\begin{example}
Consider the $[[6,4,2]]$ Iceberg code and the logical operator $\bar{P} = \bar{X}^{(0)}\bar{X}^{(1)}\bar{X}^{(2)}\bar{X}^{(3)}$.

\textit{Step 1 (physical representative).} Noting $\bar{X}^{(i)}=X_0X_{i+1}$ we have
\begin{equation*}
    P_{\mathrm{phys}} = (X_0X_1)(X_0X_2)(X_0X_3)(X_0X_4) = X_1X_2X_3X_4,
\end{equation*}
this corresponds to tuple $(2,4,0,0)$ and has weight $w=4$.

\textit{Step 2 (select the minimal weight representative).} The four representatives are
\begin{align*}
    P_{\mathrm{phys}}        &= X_1X_2X_3X_4, & t&=(2,4,0,0), & w&=4,\\
    s_X P_{\mathrm{phys}}    &= X_0X_5,       & t&=(4,2,0,0), & w&=2,\\
    s_Z P_{\mathrm{phys}}    &= Z_0Y_1Y_2Y_3Y_4Z_5, & t&=(0,0,4,2), & w&=6,\\
    s_Xs_Z P_{\mathrm{phys}} &= Y_0Z_1Z_2Z_3Z_4Y_5, & t&=(0,0,2,4), & w&=6.
\end{align*}
The minimal weight is obtained by $s_XP_{\mathrm{phys}}=X_0X_5$, so we take $P_{\mathrm{phys}}\gets X_0X_5$ and the tuple is $t=(4,2,0,0)$ with weight $w=2$.

\textit{Step 3 (read off the seed).} From $t=(4,2,0,0)$ the canonical seed is
\begin{equation*}
    P_{\mathrm{seed}} = \underbrace{X\,X}_{n_X=2}\,\underbrace{I\,I\,I\,I}_{n_I=4}
    = X_0X_1.
\end{equation*}

\textit{Step 4 (find the permutation).} Group the qubit indices by the Pauli type carried at each site of $P_{\mathrm{phys}}=X_0X_5$:
\begin{equation*}
    I_X = \{0,5\},\quad I_Y=I_Z=\emptyset,\quad I_I=\{1,2,3,4\}.
\end{equation*}
Listing these groups in the order $X,Y,Z,I$ gives the target ordering of sites; the permutation $\sigma$ taking the sites of $P_{\mathrm{phys}}$ to this order fixes $0$, sends $5\mapsto 1$, and relabels $\{1,2,3,4\}\mapsto\{2,3,4,5\}$. By construction $\sigma(X_0X_5)=X_0X_1=P_{\mathrm{seed}}$.

Thus the all-$X$ logical lies in the $XX$ orbit and admits a weight-$2$ representative. Since $\sigma$ is a relabelling, it need not be applied physically. Rather than permuting the data qubits so that $P_{\mathrm{phys}}$ matches the seed layout, one addresses the seed gadget to the permuted qubits, coupling each gauge qubit and check to the inverse permutation $\sigma^{-1}$ of its original target. On hardware with reconfigurable connectivity this is a change of addressing. The stabilisers $s_X$ and $s_Z$ are themselves permutation invariant, so the code and its syndrome data are unchanged.
\end{example}

\begin{remark}
The seed reduction \cref{alg:seed} runs in linear time in the block size. Computing the physical representative, minimising over the four stabiliser coset representatives, and counting the Pauli types each require only a single pass over the $2N$ sites. The final permutation is obtained by sorting the sites into four buckets, so the full algorithm runs in $O(N)$ time.
\end{remark}

\begin{remark}
The number of $S_{2N}$-orbits grows as $O(N^3)$, compared to the $4^{2N-2}$ logical Paulis. Hence, the number of gadgets that must be prepared and calibrated for a computation is polynomial in the block size, and any logical Pauli measurement can be compiled to one of these gadgets via the $O(N)$ seed reduction of \cref{alg:seed} and a free relabelling.
\end{remark}

\begin{algorithm}[H]
\caption{Seed reduction for the Iceberg code}
\label{alg:seed}
\begin{algorithmic}[1]
\Require A logical Pauli string $\bar{P}$
\Ensure Seed $P_{\mathrm{seed}}$ and permutation $\sigma\in S_{2N}$
  \State \textbf{Compute the physical representative.} Determine the physical
    Pauli string $P_{\mathrm{phys}}$ on $2N$ qubits whose logical action is $\bar{P}$.
\State \textbf{Minimise weight over the stabiliser coset.} Among the four
  elements $\{P_{\mathrm{phys}},\,s_XP_{\mathrm{phys}},\,s_ZP_{\mathrm{phys}},\,s_Xs_ZP_{\mathrm{phys}}\}$,
  replace $P_{\mathrm{phys}}$ by the one with minimal weight; if more than one
  achieves the minimum, then take the seed with the smallest $n_Y$.
  \State \textbf{Read off the seed.} From $t(P_{\mathrm{phys}})=(n_I,n_X,n_Y,n_Z)$,
    construct the seed string
    \begin{equation*}
        P_{\mathrm{seed}} = \underbrace{X\cdots X}_{n_X}\,\underbrace{Y\cdots Y}_{n_Y}\,\underbrace{Z\cdots Z}_{n_Z}\,\underbrace{I\cdots I}_{n_I}.
    \end{equation*}
\State \textbf{Find the permutation.} Partition the $2N$ qubit indices into four buckets $B_X,B_Y,B_Z,B_I$ according to the Pauli type carried by $P_{\mathrm{phys}}$ at each site, so that $B_P = \{\, i : (P_{\mathrm{phys}})_i = P \,\}$ for each $P\in\{X,Y,Z,I\}$. Concatenate these buckets in the order $X,Y,Z,I$ into a single ordered list $L = (L_0, L_1, \ldots, L_{2N-1}) = B_X \| B_Y \| B_Z \| B_I$, so that $L_j$ is the $j$-th site in this list. Define $\sigma$ by $\sigma(L_j) = j$ for each $j$ --- i.e.\ the site occupying list-position $j$ is sent to position $j$. This satisfies $\sigma(P_{\mathrm{phys}}) = P_{\mathrm{seed}}$, since $L$ lists the sites in the same type order as the
seed string.
  \State \textbf{Output.} Return $P_{\mathrm{seed}}$ and $\sigma$
\end{algorithmic}
\end{algorithm}

\section{Surgery}
\label{sec:gadgets}

We now construct the surgery gadgets for the Iceberg codes. Let us first recall the protocol from the warm-up example of \cref{sec:warmup}. To measure a logical Pauli $\bar P$, we add gauge qubits and Gauss-law checks---the gadget---deforming the code into a merged code in which $\bar P$ is a product of stabilisers. Measuring the merged code checks yields the eigenvalue of $\bar P$ as the product of Gauss-law outcomes, and a split returns the data to the original code. The warm up treated a single weight two operator of the $[[4,2,2]]$ code and here we provide a more general treatment. We give the gauging recipe for an arbitrary seed operator (\cref{sec:gauging-general}), explicit gadgets for every seed of the $[[4,2,2]]$ and $[[6,4,2]]$ codes, the merge--split measurement protocol, and the compilation of an arbitrary logical Pauli via the orbit classification of \cref{sec:orbits}. Circuit level numerical verification of the fault-detecting property is then presented in \cref{sec:ft}.

\subsection{Gadget construction}
\label{sec:gauging-general}

We gauge an arbitrary seed operator of the $[[2N,2N-2,2]]$ Iceberg code. Let $P = \prod_{v\in V} P_v$ be a seed of weight $w = n_X+n_Y+n_Z$, with $P_v\in\{X,Y,Z\}$ supported on the first $w$ data qubits $V = \{0,1,\dots,w-1\}$. Choose a connected auxiliary graph $\Gamma = (V,E)$ on the support, place one gauge qubit $\tau_e$ on each edge $e\in E$, and define the Gauss-law checks by
\begin{equation}
  G_v = P_v \prod_{e\ni v}\tau_e^X, \qquad v\in V.
\end{equation}
These operators mutually commute, and since every edge meets exactly two vertices the gauge factors cancel in pairs in the product $\prod_{v\in V} G_v = P$ so that the product of the Gauss-law outcomes is the measured eigenvalue of $P$.

The two Iceberg stabilisers $s_X = X^{\otimes 2N}$ and $s_Z = Z^{\otimes 2N}$ must now be dressed by Wilson lines\footnote{In the gauge theory language, the Gauss laws $G_v$ generate local $\mathbb{Z}_2$ transformations and the bare stabiliser $s_X$ becomes charged under the gauge symmetry --- conjugation by $G_v$ produces a sign at each vertex of $C_X$. A product $\prod_{e\in D}\tau^Z_e$ along a set of edges is an open Wilson line, itself charged precisely at its endpoints $\partial D$. Attaching a Wilson line whose endpoints coincide with the charged sites, $\partial D_X = C_X$, cancels the charge and ensures that the dressed operator $W_X$ is gauge invariant. The dressings are thus open Wilson lines terminating on the charges.} to commute with the Gauss laws. Since $s_X$ acts on every qubit, it anticommutes with $G_v$ exactly at those $v$ where $P_v \in \{Y,Z\}$; call this set $C_X\subseteq V$. Similarly $s_Z$ anticommutes with $G_v$ exactly where $P_v\in\{X,Y\}$, giving $C_Z$. We thus look for edge sets $D_X\subseteq E$ and $D_Z\subseteq E$ such that
\begin{equation}\label{eq:dressedstab}
  W_X = s_X\prod_{e\in D_X}\tau_e^Z, \qquad
  W_Z = s_Z\prod_{e\in D_Z}\tau_e^Z
\end{equation}
commute with every $G_v$. The matter factors contribute a sign at $v$ exactly when $v\in C_X$, the gauge factors exactly when an odd number of edges of $D_X$ meet $v$, so writing $\partial: \mathbb{F}_2^E\to\mathbb{F}_2^V$ for the boundary map of $\Gamma$, the dressing condition\footnote{Solutions of $\partial D_X = C_X$ form a coset of the cycle space, so any two valid dressings differ by a product of $\tau^Z_e$ around cycles of $\Gamma$---flux checks of the merged code, defined below---and the merged code is independent of the choice.} is $\partial D_X = C_X$ and $\partial D_Z = C_Z$.

\begin{remark}
We note that such dressings always exist for the Iceberg code. Since every edge has two endpoints, the image of $\partial$ lies in the even-weight subspace of $\mathbb{F}_2^V$; the kernel of $\partial$ is the cycle space of $\Gamma$, of dimension $|E| - w + 1$ for $\Gamma$ connected, so the image is exactly the even-weight subspace, of dimension $w-1$. Now $|C_X| = n_Y + n_Z$ and $|C_Z| = n_X + n_Y$, and both are even. Indeed, these are precisely the parity constraints of \cref{sec:orbits} that make $P$ a logical operator of the Iceberg code. 
\end{remark}

In the case that $\Gamma$ contains cycles, each independent cycle $c$ contributes a flux check
\begin{equation}
  B_c = \prod_{e\in c}\tau_e^Z,
\end{equation}
which commutes with every Gauss law since a cycle meets each vertex in an even number of edges, and the merged code is finally defined as
\begin{equation}
  \mathcal{S}_{\mathrm{merge}} = \bigl\langle\, W_X,\; W_Z,\; \{G_v\}_{v\in V},\; \{B_c\}_{c\in\mathcal{C}}\,\bigr\rangle,
\end{equation}
where $\mathcal{C}$ is a set of independent cycles of $\Gamma$, using $|E|$ gauge qubits in total. The construction is illustrated, for example, in \cref{fig:loop_check}.

It remains to choose $\Gamma$ such that the code distance (in this case $d=2$) is preserved in the merged code. The relevant quantity is the \emph{Cheeger constant}:
\begin{equation}
  h(\Gamma) := \min_{0 < |S| \le |V|/2} \frac{|\partial S|}{|S|},
\end{equation}
where the minimum is taken over vertex subsets $S$ and $\partial S$ denotes the set of edges with exactly one endpoint in $S$. The merged code preserves the distance of the original whenever $h(\Gamma)\geq 1$ \cite{williamson2026low}. For the path graph on $w$ vertices we have $h(\Gamma) = 1/\lfloor w/2\rfloor$, so the path suffices exactly for weight $w\leq 3$ operators at the minimal cost of $w-1$ gauge qubits and no flux checks. For $w > 3$ the path must be augmented. Adding edges never decreases $h(\Gamma)$, and the complete graph has $h\geq 1$, so a distance-preserving graph always exists, with each added edge costing one gauge qubit, one flux check, and higher-weight Gauss laws at its endpoints. 

\begin{example}
We consider now the $[[4,2,2]]$ example. The orbit classification of \cref{sec:orbits} gives four seeds:
\begin{equation}
    \{\, XXII,\; YYII,\; ZZII,\; XYZI \,\},
\end{equation}
all of weight $w\leq 3$, so each is gauged on the path graph over its support. The resulting merged-code generators (dressed checks $W_X, W_Z$ and Gauss laws $G_v$) are collected in \cref{tab:gadgets4} and the Tanner graph is illustrated in \cref{fig:gadget4}. The weight two seeds use a single gauge qubit and the $XYZ$ seed uses two.
\end{example}

\begin{example}
For the $[[6,4,2]]$ code the classification of  \cref{sec:orbits} gives five seeds,
\begin{equation}
    \{\, XXIIII,\; YYIIII,\; ZZIIII,\; XYZIII,\; XXZZII \,\},
\end{equation}
the first four of weight $w\leq 3$, gauged by the path graph as before. The weight four seed $XXZZII$ is the one case requiring modification since on the bare path the merged code acquires an unwanted weight one logical operator (reflected in the fact that the Cheeger constant is $h(\Gamma)=1/2$). Closing the path with a single edge $e_3=(3,0)$ into a four cycle restores $h(\Gamma)=1$. The cycle contributes the flux check $B = \tau_{e_0}^Z\tau_{e_1}^Z\tau_{e_2}^Z\tau_{e_3}^Z$ and one additional gauge qubit. The stabilisers of all five gadgets are collected in \cref{tab:gadgets6}. Many physical layouts are possible, and we illustrate one alternative example for the weight four gadget in  \cref{fig:loop_check}.
\end{example}

\begin{figure}
    \centering
    \includegraphics[width=0.35\linewidth]{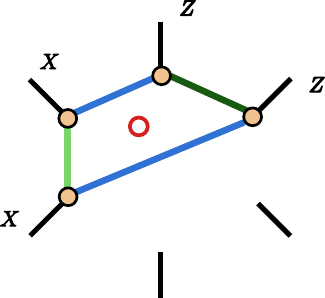}
    \caption{The auxiliary graph $\Gamma$ for the weight-four seed $XXZZII$ of the $[[6,4,2]]$ Iceberg code. The six data qubits are drawn as outward spokes of a hexagon, with the support of the seed labelled by Pauli type and the two unsupported qubits in grey. Gauge qubits live on the edges of $\Gamma$ (blue and green): the path along the hexagon is closed into a four-cycle by a chord. Gauss-law checks $G_v$ sit at the pink vertices and the flux check $B$ on the enclosed face (red). The dressings of the two global checks are marked in green: $W_Z = s_Z\,\tau^Z_{e_0}$ (light green) and $W_X = s_X\,\tau^Z_{e_2}$ (dark green).}
    \label{fig:loop_check}
\end{figure}

\begin{table}[H]
\centering
\renewcommand{\arraystretch}{1.4}
\begin{tabular}{cccccccc}
\hline\hline
$(n_I,n_X,n_Y,n_Z)$ & $P$ & $W_X$ & $W_Z$
    & $G_{v_0}$ & $G_{v_1}$ & $G_{v_2}$ & $n_g$ \\
\hline
$(2,2,0,0)$
    & $X^{\otimes 2}I^{\otimes 2}$
    & $X^{\otimes 4}$
    & $Z^{\otimes 4}\tau_{e_0}^Z$
    & $X_0\,\tau_{e_0}^X$ & $X_1\,\tau_{e_0}^X$ & & 1\\[4pt]
$(2,0,2,0)$
    & $Y^{\otimes 2}I^{\otimes 2}$
    & $X^{\otimes 4}\tau_{e_0}^Z$
    & $Z^{\otimes 4}\tau_{e_0}^Z$
    & $Y_0\,\tau_{e_0}^X$ & $Y_1\,\tau_{e_0}^X$ & & 1\\[4pt]
$(2,0,0,2)$
    & $Z^{\otimes 2}I^{\otimes 2}$
    & $X^{\otimes 4}\tau_{e_0}^Z$
    & $Z^{\otimes 4}$
    & $Z_0\,\tau_{e_0}^X$ & $Z_1\,\tau_{e_0}^X$ & & 1\\[4pt]
$(1,1,1,1)$
    & $XYZI$
    & $X^{\otimes 4}\tau_{e_1}^Z$
    & $Z^{\otimes 4}\tau_{e_0}^Z$
    & $X_0\,\tau_{e_0}^X$
    & $Y_1\,\tau_{e_0}^X\tau_{e_1}^X$
    & $Z_2\,\tau_{e_1}^X$ & 2\\
\hline\hline
\end{tabular}
\caption{Stabiliser generators $\langle W_X,\,W_Z,\,G_{v_0}\,G_{v_1}\,G_{v_2}\rangle$ of the
merged code for each seed $P$ of the $[[4,2,2]]$ Iceberg code. The last column $n_g$
gives the number of gauge qubits introduced.}
\label{tab:gadgets4}
\end{table}

\begin{table}[H]
\centering
\footnotesize
\setlength\tabcolsep{3.5pt}
\renewcommand{\arraystretch}{1.25}
\begin{tabular}{cccccccccc}
\hline\hline
$(n_I,n_X,n_Y,n_Z)$ & $P$ & $W_X$ & $W_Z$ & $G_{v_0}$ & $G_{v_1}$ & $G_{v_2}$ & $G_{v_3}$ & $B$ & $n_g$ \tabularnewline
\hline
$(4,2,0,0)$ & $X^{\otimes 2}I^{\otimes 4}$ & $X^{\otimes 6}$ & $Z^{\otimes 6}\tau_0^Z$ & $X_0\tau_0^X$ & $X_1\tau_0^X$ & -- & -- & -- & $1$ \tabularnewline
$(4,0,2,0)$ & $Y^{\otimes 2}I^{\otimes 4}$ & $X^{\otimes 6}\tau_0^Z$ & $Z^{\otimes 6}\tau_0^Z$ & $Y_0\tau_0^X$ & $Y_1\tau_0^X$ & -- & -- & -- & $1$ \tabularnewline
$(4,0,0,2)$ & $Z^{\otimes 2}I^{\otimes 4}$ & $X^{\otimes 6}\tau_0^Z$ & $Z^{\otimes 6}$ & $Z_0\tau_0^X$ & $Z_1\tau_0^X$ & -- & -- & -- & $1$ \tabularnewline
$(2,2,0,2)$ & $X^{\otimes 2}Z^{\otimes 2}I^{\otimes 2}$ & $X^{\otimes 6}\tau_2^Z$ & $Z^{\otimes 6}\tau_0^Z$ & $X_0\tau_0^X\tau_3^X$ & $X_1\tau_0^X\tau_1^X$ & $Z_2\tau_1^X\tau_2^X$ & $Z_3\tau_2^X\tau_3^X$ & $\tau_0^Z\tau_1^Z\tau_2^Z\tau_3^Z$ & $4$ \tabularnewline
$(3,1,1,1)$ & $XYZI^{\otimes 3}$ & $X^{\otimes 6}\tau_1^Z$ & $Z^{\otimes 6}\tau_0^Z$ & $X_0\tau_0^X$ & $Y_1\tau_0^X\tau_1^X$ & $Z_2\tau_1^X$ & -- & -- & $2$ \tabularnewline
\hline\hline
\end{tabular}
\caption{Stabiliser generators of the merged code for the five seed representatives of the $[[6,4,2]]$ Iceberg code. The $XXZZII$ seed has a four-cycle and requires the flux check $B$.}
\label{tab:gadgets6}
\end{table}

\begin{figure}[h]
    \centering
    \includegraphics[width=0.8\textwidth]{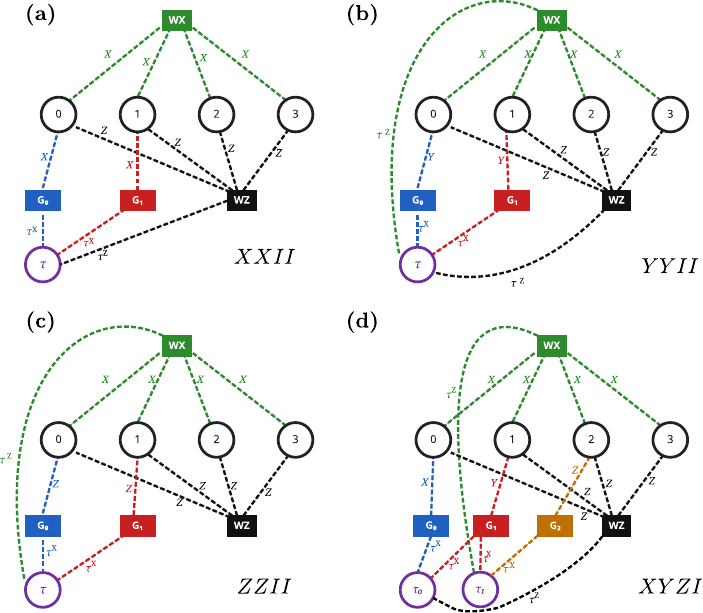}
    \caption{The panels depict the stabiliser structure of the merged $[[4,2,2]]$ Iceberg code for different seed operators. The black circles represent the four physical data qubits, and the purple circles represent the gauge qubits. The green ancilla $W_X$ measures the $X$-type stabiliser, the black ancilla $W_Z$ measures the $Z$-type stabiliser, and the colored ancillas $G_{v_0}$, $G_{v_1}$ and $G_{v_2}$ are the Gauss-law checks at
each support vertex.}
    \label{fig:gadget4}
\end{figure}

Our focus on the $[[4,2,2]]$ and $[[6,4,2]]$ codes is also experimentally motivated. Since fault detection proceeds by post-selection, the acceptance rate falls with the number of noisy circuit locations, so larger blocks and higher-weight gadgets reject an increasing fraction of runs. The small examples here keep the accepted fraction high enough to be practical.

\subsection{The measurement protocol}
We now describe the protocol that measures some seed $P$ and explain how to perform post-selected error detection. In \Cref{sec:ft} we provide numerical simulations that demonstrate the required $O(p^2)$ logical error rate of the following protocol. Additional details are discussed in \cref{app:gaugedet}.

The Iceberg code has distance two, and our surgery gadgets are fault detecting rather than fault correcting. This means that every single circuit-level fault must either act trivially or cause the run to be rejected. We are thus concerned with the post-selected logical error rate 
\begin{equation}
  P_{\mathrm{L}\mid\mathrm{acc}}(p) = \Pr \left[\,\text{logical failure}\mid\text{accepted}\,\right].
  \label{eq:plgacc}
\end{equation}

Let us assume that the Iceberg code has already been prepared in the code space state\footnote{For instance by the fault-tolerant, error-detecting initialisation of \cite{self2024protecting}.} $\ket{\psi}$ and that we have the last measurements of the two stabilisers $s_X$ and $s_Z$ of the previous rounds.

\paragraph{Initialisation.}
In the first step of the protocol we reset/introduce the gauge qubits in the $\ket{0}$ state, so that the system is in the state $\ket{\psi}\ket{0}^{\otimes n_g}$. This state is the stabiliser state of the split code 
\begin{equation}
    \mathcal{S}_{\mathrm{split}} = \langle s_X,s_Z,\{\tau^Z_e\} \rangle.
\end{equation}
Throughout this section we write $s(t)\in\{0,1\}$ for the outcome of the $t$-th measurement of a stabiliser operator $s$, with the convention $0\leftrightarrow+1$ and $1\leftrightarrow-1$ for the corresponding eigenvalue. Comparison detector outcomes are computed as the mod $2$ sum ($\oplus$) of two such measurement outcomes, and a detector signals a fault whenever it evaluates to $1$, i.e. whenever the two measurements it compares disagree.

We prepare $\ket{\psi}\ket{0}^{\otimes n_g}$ fault-tolerantly with two rounds of split-code syndrome extraction, defining the detectors
\begin{gather}
    D_{0}=s_{X}(0)\oplus s_{X}(-1),\quad D_{1}=s_{Z}(0)\oplus s_{Z}(-1),\quad D_{2}^{(e)}=\tau_{e}^Z(0),\\
    D_{3}=s_{X}(1)\oplus s_{X}(0),\quad D_{4}=s_{Z}(1)\oplus s_{Z}(0),\quad D_{5}^{(e)}=\tau_{e}^Z(1)\oplus \tau_{e}^Z(0),
\end{gather}
where $t=-1$ labels the last measurement recorded before this protocol began and $t=0,1$ label the two rounds of this step. In the first round, $D_0,D_1$ check $s_X,s_Z$ against their previously recorded values, while $D_2^{(e)}$ checks that each gauge qubit was reset successfully, i.e. measured as $+1$; unlike $s_X,s_Z$, the gauge qubits have no prior round to compare against, so this detector instead checks the expected fixed value directly. In the second round, $D_3,D_4,D_5^{(e)}$ compare all outcomes against the first round. Together these detectors flag any fault occurring during the preparation of the split code state.

\paragraph{Merging.} 
Having prepared the split code state, we now measure the merged code stabilisers
\begin{equation}
    \mathcal{S}_{\mathrm{merge}} = \langle W_X, W_Z, \{B_c\}, \{G_v\}\rangle
\end{equation}
for two rounds, continuing the round labels at $t=2,3$. This step is where the measurement outcome $\mu$ is extracted from the $G_v$ outcomes.

The merged code stabilisers $W_X$ and $W_Z$ are given by $W_X = s_X\prod_{e\in D_X}\tau_e^Z$ and $W_Z=s_Z\prod_{e\in D_Z}\tau_e^Z$, for fixed edge sets $D_X$ and $D_Z$ determined by the gadget. Their first round outcome is therefore predictable from the already known $s_X,s_Z$ together with the gauge qubit values measured at the end of the split code stage (round $1$), giving the consistency detectors
\begin{equation}
    D_{6} = W_{X}(2)\oplus s_{X}(1)\oplus\bigoplus_{e\in D_{X}}\tau_{e}^Z(1), \qquad
    D_{7} = W_{Z}(2)\oplus s_{Z}(1)\oplus\bigoplus_{e\in D_{Z}}\tau_{e}^Z(1).
\end{equation}
Similarly, the plaquette $B_c=\prod_{e\in c}\tau_{e}^Z$ is built entirely out of gauge qubits, so its first round outcome should match the corresponding combination of split code gauge qubit outcomes,
\begin{equation}
    D_{8}^{(c)} = B_{c}(2)\oplus\bigoplus_{e\in c}\tau_{e}^Z(1), \qquad \text{one detector per cycle } c.
\end{equation}
$G_v$ was not part of the split-code stabiliser group, so there is nothing from the previous stage to compare it against. Its first measurement outcome $g_v(2)$ is non-deterministic ($\langle G_v\rangle=0$ on the split-code state), and no fault detecting comparison detector can be defined for it yet.

In the second round, we repeat the measurement of every merged code stabiliser, comparing round $3$ against round $2$,
\begin{equation}
    D_{9} = W_{X}(3)\oplus W_{X}(2),\qquad
    D_{10} = W_{Z}(3)\oplus W_{Z}(2),\qquad
    D_{11}^{(c)} = B_{c}(3)\oplus B_{c}(2).
\end{equation}
For $G_v$ this is the first round in which a detector is possible:
\begin{equation}
    D_{12}^{(v)} = G_{v}(3)\oplus G_{v}(2).
\end{equation}
If none of these detectors fire, the round $2$ outcomes are correct, and we record
\begin{equation}
    \mu = \prod_{v} g_{v}(2),
\end{equation}
which is the measurement outcome of the logical operator $P$\footnote{Here $g_v = \pm 1$.}. The system is left in an entangled state.

\paragraph{Splitting.}
The final step disentangles the gauge qubits from the matter by measuring them directly. To protect against a measurement error we measure every gauge qubit twice, continuing the round labels at $t=4,5$, and define the detectors
\begin{equation}
    D_{13}^{(e)} = \tau_{e}^Z(5)\oplus \tau_{e}^Z(4).
\end{equation}
If none of these fire, we record the trusted outcome $w_e = \pm1$ from round $4$ for every edge $e$, and set
\begin{equation}
    m_e = \frac{1-w_e}{2} \in\{0,1\}.
\end{equation}
The bit string $m=(m_e)_{e\in E}$ must lie in the cut space of the graph for the inversion below to be possible. This is not implied by the detectors defined so far since the $m_e$ are fresh measurements, and nothing yet relates them to the flux checks. Since $B_c$ is built out of gauge qubits, its last merged code outcome must agree with the product of the split outcomes around that cycle, which we impose as the detector
\begin{equation}
    D_{14}^{(c)} = B_{c}(3)\oplus\bigoplus_{e\in c} m_e, \qquad \text{one detector per cycle } c.
    \label{eq:fluxsplit}
\end{equation}
This is precisely the statement that $m=\delta(A)$ for some $A\subseteq V$, where $\delta$ is the coboundary (cut) map sending a vertex set $A$ to the set of edges with exactly one endpoint in $A$.

We recover $A$ (up to an $A \to A^{c}$ ambiguity) by inverting $\delta$. We fix a spanning tree $T$ and a root $r$ of $G$, and compute the potential $\varphi_v = \bigoplus_{e\in\Pi_v} m_e$ along the unique tree path $\Pi_v$ from $r$ to each vertex $v$, and set $A=\{v : \varphi_v=1\}$. The byproduct Pauli correction $P_A^{-1}=P_A$, with
\begin{equation}
    P_A = \prod_{v\in A} P_v,
\end{equation}
is then applied as a classical frame update, together with $\mu$ (already known from the previous step) to resolve the residual $A\leftrightarrow A^c$ ambiguity. This returns the matter to the state $\ket{\psi_\mu}$, completing the measurement of $P$.

\paragraph{Final round.}
Finally, we perform two more rounds of split code syndrome extraction, measuring only $s_X$ and $s_Z$ (the gauge qubits have already been consumed by the previous step), at rounds $t=6,7$. This re-establishes the Iceberg code stabilisers.

In the first round, since $W_X = s_X\prod_{e\in D_X}\tau_e^Z$ and $W_Z = s_Z\prod_{e\in D_Z}\tau_e^Z$, the newly measured $s_X,s_Z$ should be consistent with the last merged code outcome $W_X(3),W_Z(3)$ together with the gauge qubit outcomes $m_e$ recorded during the splitting step, giving the consistency detectors
\begin{equation}
    D_{15} = s_{X}(6)\oplus W_{X}(3)\oplus\bigoplus_{e\in D_{X}} m_e, \qquad
    D_{16} = s_{Z}(6)\oplus W_{Z}(3)\oplus\bigoplus_{e\in D_{Z}} m_e.
\end{equation}
In the second round we compare against the first,
\begin{equation}
    D_{17} = s_{X}(7)\oplus s_{X}(6), \qquad D_{18} = s_{Z}(7)\oplus s_{Z}(6).
\end{equation}

If none of these detectors fire, the protocol is accepted: the matter is in the state $\ket{\psi_\mu}$ (up to the classically-tracked correction $P_A$), and $s_X(7),s_Z(7)$ are the trusted stabiliser outcomes to be carried forward into subsequent rounds.

We note that extraction of a high weight check --- $s_X,s_Z$ in the base Iceberg code, or $W_X,W_Z$ in the merged code --- is not always fault tolerant. We therefore measure these checks with a standard flagged syndrome extraction \cite{chao2018fault,chamberland2018flag}, so an ancilla hook flips the flag and in that case the run is discarded.

\subsubsection{General logical measurement protocol}
We now state the full procedure for measuring an arbitrary logical Pauli $\bar P$, combining the seed reduction of \Cref{sec:orbits} with the gadget protocol above. Given $\bar P$:
\begin{enumerate}
  \item \textit{Reduce to a seed.} Run \Cref{alg:seed} to obtain the seed
    $P_{\mathrm{seed}}$ of the orbit of $\bar P$ and the permutation
    $\sigma\in S_{2N}$ with $\sigma(P_{\mathrm{phys}})=P_{\mathrm{seed}}$.
  \item \textit{Relabel.} Apply $\sigma$ as a relabelling of the physical qubits.
  \item \textit{Measure the seed gadget.} Run the three phase merge--split protocol for the gadget of $P_{\mathrm{seed}}$ (\Cref{tab:gadgets4,tab:gadgets6}), obtaining the outcome $\mu=\prod_v g_v$ and post-selecting on the detection checks.
  \item \textit{Undo the relabelling.} Invert $\sigma$ to return to the original
    labelling. Record $\mu$, the eigenvalue of $\bar P$.
\end{enumerate}
Since the orbit classification provides one gadget per seed and $\sigma$ is free, a
small (polynomial size in $N$) collection of gadgets --- four for the $[[4,2,2]]$, five for the
$[[6,4,2]]$ --- measures every logical Pauli of the Iceberg code.

\subsection{Numerical simulation}
\label{sec:ft}

We simulate the merge/split protocol of the previous section at circuit level using \textsc{Stim}~\cite{gidney2021stim} and confirm the expected $O(p^2)$ suppression of the post-selected logical error rate \eqref{eq:plgacc}.

\paragraph{Noise model.}

Every single qubit gate is followed by one qubit depolarising noise of strength $p$, every two qubit gate by two qubit depolarising noise of strength $p$, every reset by a bit flip with probability $p$, and every measurement reports the wrong outcome with probability $p$.

\paragraph{Observables and acceptance.}
The protocol must correctly report $\mu$, the true eigenvalue of $P$, while leaving every other logical untouched. The surviving logical algebra of the merged code has $k=2N-3$ qubits (\Cref{sec:gauging-general}), forming a $2k$-dimensional symplectic space in which $P$ is one nonzero vector. A Lagrangian completion of $P$ is a maximal commuting set of logical operators containing $P$, of size $k+1$. To certify that every residual logical direction is protected, not just $P$ itself, we run the protocol twice, each time reading out a different Lagrangian completion $L_A,L_B$ of $P$, chosen so that $L_A$ and $L_B$ together span the full symplectic space. Any undetected logical error would then have to commute with everything in both $L_A$ and $L_B$, and hence be trivial. Each run exports $k+1$ observables checking the reported eigenvalue, the preservation of the other $k$ logicals in that run's completion, and a final readout of $P$. As a direct check of the fault-detecting property, we also enumerate all single circuit-level faults and confirm that each is either detected or acts trivially on the logical state.

\paragraph{Results.}

We compute $p\in[10^{-3},7\times10^{-3}]$ at seven logarithmically spaced points for every seed gadget of the $[[4,2,2]]$ and $[[6,4,2]]$ codes, in both runs. We adjust the number of samples at each point until the relative standard error falls below $2.5\%$, requiring between $2\times10^6$ and $5.8\times10^7$ shots per point.

\begin{table}[t]
  \centering
  \begin{tabular}{llcccc}
    \toprule
    & & \multicolumn{2}{c}{run $A$} & \multicolumn{2}{c}{run $B$} \\
    \cmidrule(lr){3-4}\cmidrule(lr){5-6}
    Code & Seed & slope & $C$ & slope & $C$ \\
    \midrule
    $[[4,2,2]]$ & $XXII$   & $2.02$ & $32.4$ & $2.01$ & $33.5$ \\
                & $YYII$   & $2.01$ & $36.1$ & $2.01$ & $33.8$ \\
                & $ZZII$   & $2.02$ & $32.5$ & $2.03$ & $32.2$ \\
                & $XYZI$   & $2.02$ & $54.7$ & $2.00$ & $55.3$ \\
    \midrule
    $[[6,4,2]]$ & $XXIIII$ & $2.00$ & $41.5$ & $2.04$ & $39.2$ \\
                & $YYIIII$ & $2.03$ & $49.0$ & $2.02$ & $44.2$ \\
                & $ZZIIII$ & $2.01$ & $38.8$ & $2.00$ & $41.8$ \\
                & $XYZIII$ & $2.02$ & $63.5$ & $2.01$ & $64.3$ \\
                & $XXZZII$ & $2.04$ & $80.4$ & $2.01$ & $74.7$ \\
    \bottomrule
  \end{tabular}
  \caption{Fitted log--log slopes and coefficients
    ($P_{\mathrm{L}\mid\mathrm{acc}}\simeq Cp^2$) for every seed gadget of the two codes.}
  \label{tab:pl}
\end{table}

\begin{figure}[t]
  \centering
  \includegraphics[width=\textwidth]{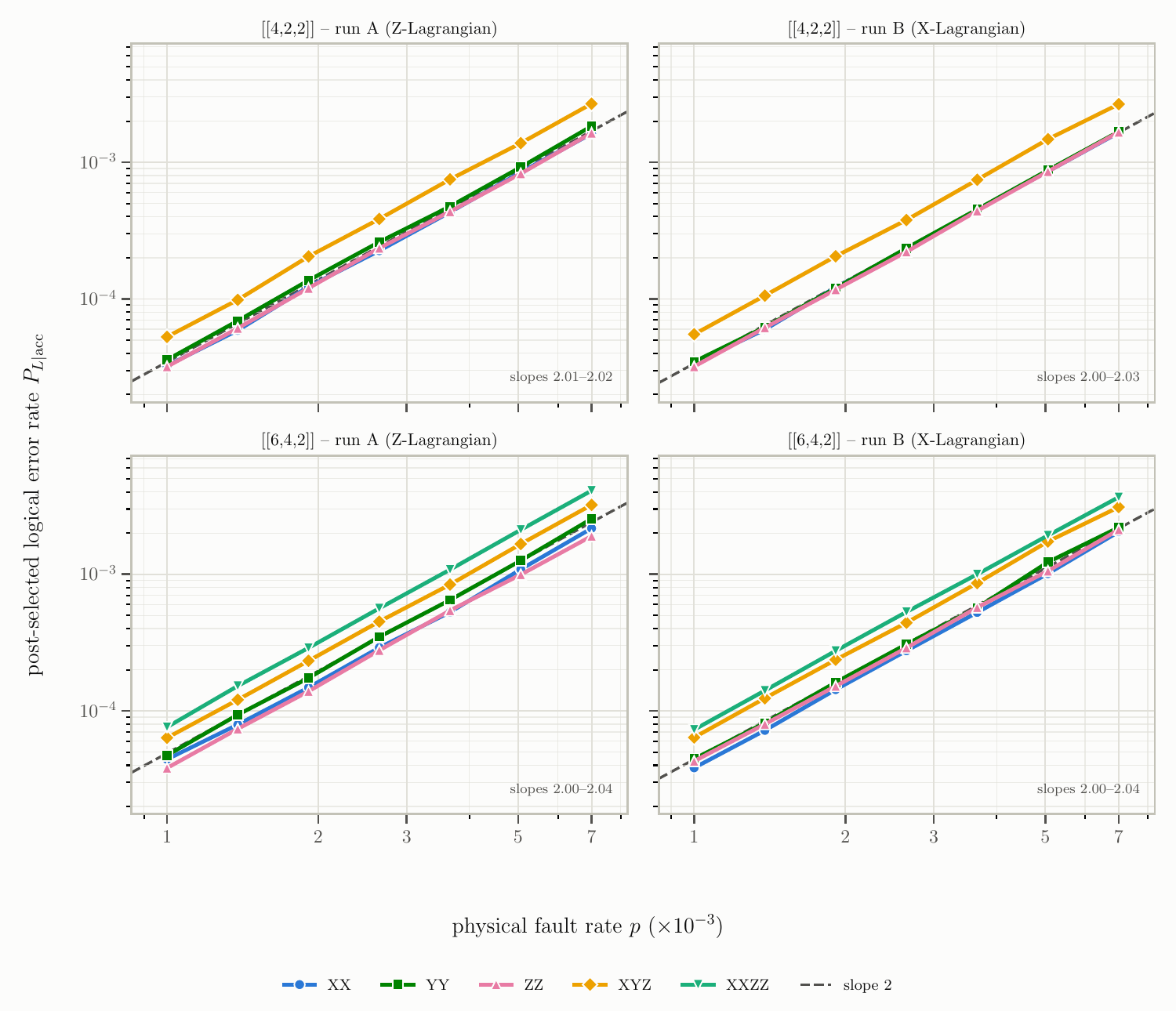}
  \caption{Post-selected logical error rate for the $[[4,2,2]]$ (top) and $[[6,4,2]]$ (bottom) seed gadgets, run $A$ (left) and run $B$ (right). Error bars are smaller than the marker size.}
  \label{fig:pl}
\end{figure}

Fitted slopes lie between $2.00$ and $2.04$ across all eighteen curves (\Cref{tab:pl}), confirming the expected quadratic suppression, with the mixed seed $XYZ$ sitting roughly $1.5$--$1.6$ times higher than the pure-type seeds, consistent with its larger gadget---two gauge qubits instead of one, and correspondingly more fault locations. The same ordering by gadget size continues in the $[[6,4,2]]$ code, where the weight four $XXZZII$ seed, with four gauge qubits and a flux check, is the most expensive.

The discard fraction grows from $11$--$13\%$ at $p=10^{-3}$ to $55$--$62\%$ at $p=7\times10^{-3}$ for the $[[4,2,2]]$ gadgets, and from $12$--$17\%$ to $60$--$73\%$ for the $[[6,4,2]]$ ones. This is the post-selection cost, explains our focus on the small Iceberg examples discussed here. These acceptance rates are shown in \Cref{fig:accept}.

\begin{figure}[t]
  \centering
  \includegraphics[width=\textwidth]{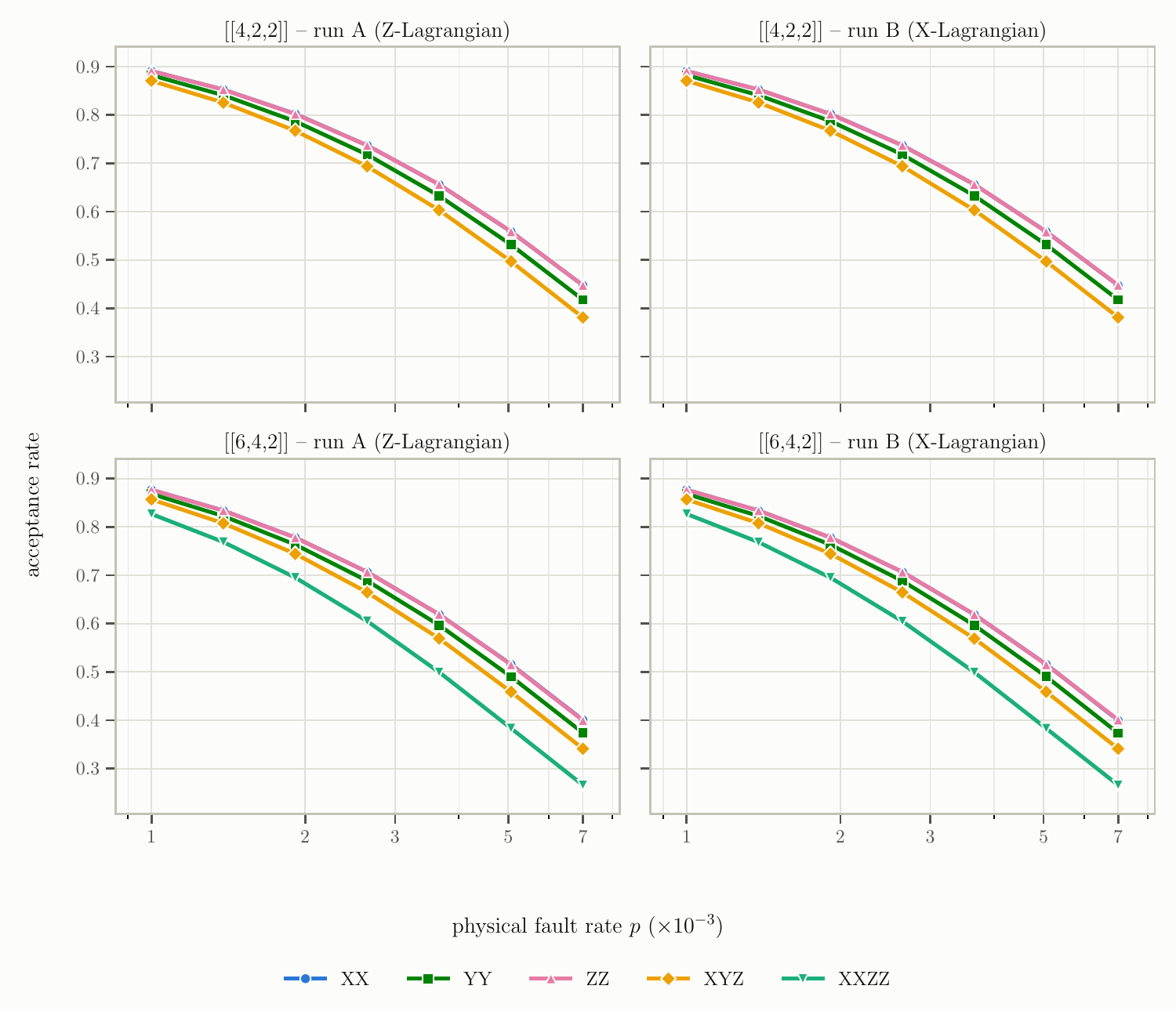}
  \caption{Acceptance rate for the $[[4,2,2]]$ (top) and $[[6,4,2]]$ (bottom) seed gadgets, run $A$ (left) and run $B$ (right). The pure type seeds $XX,YY,ZZ$ are accepted most often, followed by $XYZ$ and, in the $[[6,4,2]]$ code, the weight-four $XXZZ$ seed. Error bars are smaller than the marker size.}
  \label{fig:accept}
\end{figure}

\subsection{Examples}
We conclude the section with a number of additional examples and remarks. 

\begin{example}
We look back to Pauli-based computation of \Cref{sec:pbc-step} and realise the four-qubit GHZ sampling circuit on a single $[[6,4,2]]$ block as an example. Recall
that the circuit compiles to the Pauli product measurements
\begin{equation}
  \mathcal{P}_1=\bar{X}_0,\quad
  \mathcal{P}_2=\bar{X}_0\bar{Z}_1,\quad
  \mathcal{P}_3=\bar{X}_0\bar{Z}_1\bar{Z}_2,\quad
  \mathcal{P}_4=\bar{X}_0\bar{Z}_1\bar{Z}_2\bar{Z}_3,
\end{equation}
to be measured in order, with outcomes $m_1,\dots,m_4$ giving the GHZ samples up to Pauli frame.

\emph{Seeds and gadgets.} Using $\bar{X}^{(i)}=X_0X_{i+1}$ and
$\bar{Z}^{(i)}=Z_{i+1}Z_5$, the physical representatives are
\begin{equation}
  P_1=X_0X_1,\quad
  P_2=X_0X_1Z_2Z_5,\quad
  P_3=X_0X_1Z_2Z_3,\quad
  P_4=X_0X_1Z_2Z_3Z_4Z_5,
\end{equation}
Computing \Cref{alg:seed} on each gives
\begin{align}
  P_1 &\;\longmapsto\; P_1^{\mathrm{seed}}=X_0X_1, & \sigma_1&=\mathrm{id}, \\
P_2 &\;\longmapsto\; P_2^{\mathrm{seed}}=X_0X_1Z_2Z_3, & \sigma_2&=(3\;4\;5), \\
  P_3 &\;\longmapsto\; P_3^{\mathrm{seed}}=X_0X_1Z_2Z_3, & \sigma_3&=\mathrm{id}, \\
  P_4 &\;\longmapsto\; P_4^{\mathrm{seed}}=Y_0Y_1, & \sigma_4&=\mathrm{id}.
\end{align}
The measurement circuit therefore uses only three of the five $[[6,4,2]]$ gadgets: $XX$, $XXZZ$, and $YY$.

\emph{Execution.} Each $\mathcal{P}_i$ is measured by its gadget via the protocol above, giving $m_i\in\{0,1\}$ from the Gauss-law product $\mu_i$. As in \Cref{sec:pbc-step}, the residual Clifford is tracked classically as a Pauli frame and applied to the recorded outcomes. The whole computation runs on one $[[6,4,2]]$ block plus at most four gauge qubits (the $XXZZ$ gadget being the largest).
\end{example}

\begin{remark}[Surgery as gauge fixing of subsystem code]
The protocol has an interpretation as changing gauge of a subsystem code, as elucidated in \cite{vuillot2019code} for the surface code example. As an example, consider the $XXIIII$ gadget. Write $P=X_0X_1$, $\tau^\alpha=\tau_{e_0}^\alpha$, $W_X=s_X$, $W_Z=s_Z\tau^Z$, $G_0=X_0\tau^X$, and $G_1=X_1\tau^X$. These operators generate the parent subsystem code with gauge group
\begin{equation}
  \mathcal{G}_{\mathrm{par}}=\langle W_X,W_Z,G_0,G_1,\tau^Z\rangle,
\end{equation}
whose centre is
\begin{equation}
  Z(\mathcal{G}_{\mathrm{par}})=\langle W_X,W_Z,P\rangle,\qquad P=G_0G_1.
\end{equation}
Thus the logical Pauli being measured is a central operator of the parent subsystem code.

The initial split code $\langle s_X,s_Z,\tau^Z\rangle$ contains both $P=\pm1$ sectors, with $P$ still a logical operator. If the Gauss law measurements give outcomes $g_0,g_1\in\{\pm1\}$, the merged code is
\begin{equation}
  \mathcal{S}_{\mathrm{merge}}(g_0,g_1)=\langle W_X,W_Z,g_0G_0,g_1G_1\rangle,\qquad \mu=g_0g_1.
\end{equation}
Since $(g_0G_0)(g_1G_1)=\mu P$, these measurements both select the $P=\mu$ sector and fix the gauge. The individual outcomes $g_0$ and $g_1$ specify the gauge choice, while their product $\mu$ is the logical measurement outcome.

Within the fixed-$\mu$ sector, $G_1=\mu G_0$, so $G_0$ and $\tau^Z$ form a single anticommuting gauge pair, with subsystem stabiliser group
\begin{equation}
  \mathcal{S}_{\mathrm{par}}^{(\mu)}=\langle W_X,W_Z,\mu P\rangle.
\end{equation}
Measuring $\tau^Z$ during the return split, tracking the appropriate correction, and discarding the gauge qubit returns the data to
\begin{equation}
  \langle s_X,s_Z,\mu P\rangle.
\end{equation}
Thus, conditioned on $\mu$, the merged code and the returned split code are different gauge fixings of the same subsystem code.
\end{remark}

\begin{example}[Inter-block $XX$ measurement]
Let $A$ and $B$ be two $[[6,4,2]]$ Iceberg blocks. Label their physical qubits by $A_0,\ldots,A_5$ and $B_0,\ldots,B_5$, or equivalently by $0,\ldots,5$ and $6,\ldots,11$, respectively. We measure the product of the first logical $X$ operators,
\begin{equation}
  \bar X_A^{(0)}\bar X_B^{(0)}=X_{A,0}X_{A,1}X_{B,0}X_{B,1}=X_0X_1X_6X_7.
\end{equation}
Begin with the one edge $XXIIII$ gadget on each block, with gauge qubits $\tau_{e_0}$ and $\tau_{e_6}$, and introduce one additional bridge gauge qubit $\tau_b$ joining the two gadgets. On the bare path $0-1-6-7$ this graph has Cheeger constant $h(\Gamma)=1/2$. We close the path into a four cycle with a chord $c=(0,7)$, restoring $h(\Gamma)=1$. The Gauss-law generators are
\begin{equation}
  G_0=X_0\tau_{e_0}^X\tau_c^X,\qquad
  G_1=X_1\tau_{e_0}^X\tau_b^X,\qquad
  G_6=X_6\tau_b^X\tau_{e_6}^X,\qquad
  G_7=X_7\tau_{e_6}^X\tau_c^X.
\end{equation}
Their product is the desired inter-block logical Pauli,
\begin{equation}
  G_0G_1G_6G_7=X_0X_1X_6X_7=\bar X_A^{(0)}\bar X_B^{(0)}.
\end{equation}
The dressed stabilisers are simply those of the two single block gadgets,
\begin{equation}
  W_X^{A}=s_X^{A},\qquad W_Z^{A}=s_Z^{A}\tau_{e_0}^Z,\qquad W_X^{B}=s_X^{B},\qquad W_Z^{B}=s_Z^{B}\tau_{e_6}^Z.
\end{equation}
and the cycle contributes a flux check $B=\tau_{e_0}^Z\tau_b^Z\tau_{e_6}^Z\tau_c^Z$. The merged code is therefore generated by
\begin{equation}
  \mathcal{S}_{\mathrm{merge}}=\langle W_X^{A},W_Z^{A},W_X^{B},W_Z^{B},G_0,G_1,G_6,G_7,B\rangle,
\end{equation}
using four gauge qubits in total. If the Gauss-law outcomes are $g_0,g_1,g_6,g_7\in\{\pm1\}$, then the measurement outcome is
\begin{equation}
  \mu=g_0g_1g_6g_7,
\end{equation}
which is the eigenvalue of $\bar X_A^{(0)}\bar X_B^{(0)}$.

\end{example}

\section{Discussion}
\label{sec:discussion}

We have constructed explicit surgery gadgets for logical Pauli measurements in the $[[2N,2N-2,2]]$ Iceberg codes. The full permutation automorphism group of the Iceberg code reduces the compilation problem to one gadget per orbit seed. Thus we show that a logical Pauli can be measured by reducing it to a minimal-weight seed, performing the corresponding surgery gadget, and updating the qubit labels in software. This gives a toolkit for error-detected Pauli-based computation on Iceberg blocks. Our circuit-level simulations of the $[[4,2,2]]$ and $[[6,4,2]]$ seed gadgets exhibit the expected quadratic scaling of the accepted logical error probability under ciruit level Pauli noise.

The present work is not a proposal for an asymptotic fault-tolerant architecture since the Iceberg code has distance two, and error suppression is a consequence of post-selection. Its value is instead twofold. Firstly, we provide a small setting in which the complete gauging-based surgery primitive is small enough to be a natural target for a present-day hardware demonstration. Secondly, we present a pedagogical minimal example in which the often implicit details of surgery --- the choice of auxiliary graph, dressing of stabilisers, merge--split protocol numerics, and logical compilation by automorphisms --- can
all be worked out concretely.

\section*{Acknowledgements}

SC thanks QPerfect colleagues Hugo Perrin, Asier Pi\~neiro Orioli, Tom Hartweg, Adrian Aasen, and Shannon Whitlock for many interesting discussions.
Work at the University of Strasbourg was supported by the French National Research Agency under
the Investments of the Future Program projects ANR-21-ESRE-0032 (aQCess), ANR-17-EURE-0024 (QMat), and
ANR-22-CMAS-0001 France 2030 (QuanTEdu-France), ANR-23-CE30-50022-02 (SIX), and by the European Union via the Interreg Oberrhein programme in the project Quantum Valley Oberrhein (UpQuantVal).

\appendix

\section{Orbits and seed operators}\label{app:orbits}

The permutation automorphism group $S_{2N}$ acts on the logical Paulis of the Iceberg code, partitioning them into orbits. In this appendix we show that the orbits are in bijection with certain tuples and identify a low-weight representative in each orbit.

To a physical Pauli string $P$ we assign a tuple
\begin{equation}
  t(P)=(n_I,n_X,n_Y,n_Z),
\end{equation}
counting how many of the $2N$ sites carry each Pauli type. A permutation rearranges letters, so $t$ is a permutation invariant. Conversely $S_{2N}$ acts transitively on the set of all strings with a fixed tuple.

A tuple $t$ is realised by a logical Pauli iff its strings lie in $N(\mathcal S)\setminus\mathcal S$. Commutation with $s_X=X^{\otimes 2N}$ and $s_Z=Z^{\otimes 2N}$ requires an even overlap in each case,
\begin{equation}\label{eq:parity}
  n_X+n_Y\equiv 0,\qquad n_Z+n_Y\equiv 0 \pmod 2,
\end{equation}
together with the sum constraint
\begin{equation}\label{eq:sum}
  n_I+n_X+n_Y+n_Z=2N .
\end{equation}
Excluding $\mathcal S$ itself means excluding the tuples
\begin{equation}\label{eq:pure}
  (2N,0,0,0),\quad(0,2N,0,0),\quad(0,0,2N,0),\quad(0,0,0,2N),
\end{equation}
which are $I,\,s_X,\,s_Xs_Z,\,s_Z$. We write $\mathcal T$ for the set of tuples satisfying these constraints.

Multiplying a string by a stabiliser yields another string with the same logical action but a permuted tuple:
\begin{equation}\label{eq:V4action}
  s_X:\,t\mapsto(n_X,n_I,n_Z,n_Y),\quad
  s_Z:\,t\mapsto(n_Z,n_Y,n_X,n_I),\quad
  s_Xs_Z:\,t\mapsto(n_Y,n_Z,n_I,n_X).
\end{equation}
Writing $V_4 = \{I, s_X, s_Z, s_Xs_Z\}$ for the stabiliser group, these transformations define an action of $V_4$ on composition tuples. We denote the orbit associated to the tuple $t$ by $[t]_{V_4}$. By the orbit-stabiliser theorem we have,
\begin{equation}
  \bigl\lvert[t]_{V_4}\bigr\rvert=\frac{4}{\lvert\mathrm{Stab}(t)\rvert}.
\end{equation}
Transformations \eqref{eq:V4action} leave $t$ invariant, respectively, when one of the following conditions is satisfied.
\begin{equation}\label{eq:fixconds}
  s_X:\ n_I=n_X,\ n_Y=n_Z;\qquad
  s_Z:\ n_I=n_Z,\ n_X=n_Y;\qquad
  s_Xs_Z:\ n_I=n_Y,\ n_X=n_Z .
\end{equation}
Since if any two of \eqref{eq:fixconds} hold, the third follows, we have that $\lvert\mathrm{Stab}(t)\rvert\in\{1,2,4\}$ and correspondingly $\lvert[t]_{V_4}\rvert\in\{4,2,1\}$ respectively when no condition is satisfied, when one is satisfied, and when they are all satisfied.

\begin{proposition}\label{thm:bijection}
The $S_{2N}$-orbits of Hermitian logical Paulis (up to global sign) are in bijection with the $V_4$-orbits of valid composition tuples, via
\begin{equation}
  \varphi:\ \mathcal L/S_{2N}\longrightarrow \mathcal T/V_4,
  \qquad \varphi\bigl([P]_{S_{2N}}\bigr)=[t(P)]_{V_4}.
\end{equation}
\end{proposition}

\begin{proof}
\emph{Well-defined.} If $[P]_{S_{2N}}=[Q]_{S_{2N}}$ then $\sigma(P)=s\cdot Q$ for some $\sigma\in S_{2N}$ and $s\in V_4$. Permutations preserve tuples, so $t(P)=t(\sigma(P))=t(s\cdot Q)=V_4^{\,s}\bigl(t(Q)\bigr)$ --- where $V_4^{\,s}(\cdot)$ denotes the tuple transformation \eqref{eq:V4action} associated to $s$ --- hence $[t(P)]_{V_4}=[t(Q)]_{V_4}$.

\emph{Injective.} If $[t(P)]_{V_4}=[t(Q)]_{V_4}$ then
$t(Q)=V_4^{\,s}(t(P))=t(s\cdot P)$ for some $s\in V_4$. Since $S_{2N}$ acts transitively on strings of a fixed tuple, there is $\sigma$ with $\sigma(s\cdot P)=Q$; and since stabilisers are permutation-invariant, $\sigma(s\cdot P)=s\cdot\sigma(P)$. Thus $Q=s\cdot\sigma(P)$, giving $[Q]_{S_{2N}}=[P]_{S_{2N}}$.

\emph{Surjective.} Any valid tuple $t$ is realised by a valid string $P$ with $\varphi([P]_{S_{2N}})=[t]_{V_4}$.
\end{proof}

In addition we have, 
\begin{corollary}\label{cor:orbitsize}
For a fixed valid tuple $t$, the number of logical Paulis in the corresponding $S_{2N}$-orbit is
\begin{equation}
  \bigl\lvert[t]_{S_{2N}}\bigr\rvert
  =\frac{M(t)}{\lvert\mathrm{Stab}(t)\rvert},
  \qquad
  M(t)=\frac{(2N)!}{n_I!\,n_X!\,n_Y!\,n_Z!}.
\end{equation}
\end{corollary}

\begin{proof}
    Let us consider the valid tuple $t=(n_I,n_X,n_Y,n_Z)$. The number of all possible Pauli strings we can write that have this tuple is precisely given by $M(t)$. Among these however there may be some which have the same logical action, i.e. are related to each other by multiplication of stabilisers. This happens if $\lvert\mathrm{Stab}(t)\rvert>1$, and the number of physical strings sharing that logical action is precisely $\lvert\mathrm{Stab}(t)\rvert$.
\end{proof}
By \cref{thm:bijection}, to count the number of non-trivial $S_{2N}$-orbits, i.e. orbits not equivalent to the logical identity, it suffices to count $V_4$-orbits of valid tuples. Burnside's lemma gives
\begin{equation}
  \#\text{orbits}=\frac14\sum_{g\in V_4}\bigl\lvert\mathrm{Fix}(g)\bigr\rvert .
\end{equation}

\paragraph{$\mathrm{Fix}(e)$} The parity-consistent tuples summing to $2N$ are either all even, $n_\alpha=2a_\alpha$ with $\sum_\alpha a_\alpha=N$, or all odd, $n_\alpha=2a_\alpha+1$ with $\sum_\alpha a_\alpha=N-2$. There are $\binom{N+3}{3}$ and $\binom{N+1}{3}$ such tuples, respectively. Among the all-even tuples, exactly the four pure tuples \eqref{eq:pure} are invalid; they form a single $V_4$-orbit (trivial stabiliser) and are removed. Hence
\begin{equation}
  \lvert\mathrm{Fix}(e)\rvert=\binom{N+3}{3}+\binom{N+1}{3}-4 .
\end{equation}

\paragraph{$\mathrm{Fix}(s_X),\ \mathrm{Fix}(s_Z),\ \mathrm{Fix}(s_Xs_Z)$} A tuple fixed by $s_X$ has $n_I=n_X=a$, $n_Y=n_Z=b$ with $a+b=N$. Common parity forces $a\equiv b\pmod 2$, whereas $a+b=N$ forces $a\not\equiv b$ when $N$ is odd. Thus these fixed sets are empty for odd $N$, and for even $N$ each contains the $N+1$ tuples with $a=0,\dots,N$ (none of which is pure). The same holds for $s_Z$ and $s_Xs_Z$ by symmetry.

We have then that 
\begin{itemize}
    \item $N$ odd \begin{equation}
         \#\text{orbits} =\frac14\Bigl[\tbinom{N+3}{3}+\tbinom{N+1}{3}-4\Bigr]
    \end{equation}
    \item $N$ even \begin{equation}
          \#\text{orbits}  =\frac14\Bigl[\tbinom{N+3}{3}+\tbinom{N+1}{3}-4+3(N+1)\Bigr]
    \end{equation}
\end{itemize}

\section{Measurement protocol}\label{app:gaugedet}
In this appendix, we explore in detail the protocol from the perspective of Hamiltonian lattice gauge theory. Let us consider the generic connected graph $G=(V,E)$ for the gadget of the operator $P=\prod_{v\in V} P_v$. This graph defines a gauge theory coupling the gauge qubits living on the edges of the graph with the matter qubit of the original code. Define the vertex operators 
\begin{equation}
    G_v =P_v \prod_{e\ni v
} \tau_e^{X} \quad \forall v\in V
\end{equation}
and plaquette operators 
\begin{equation}
    B_c = \prod_{e\in p} \tau_e^Z
\end{equation}
for each independent cycle of the graph. Together with these two operators, gauging the theory dresses the hamiltonian to $H=-W_X-W_Z$ where the dressed stabilisers are given in equation \eqref{eq:dressedstab}.

The protocol starts with the matter system in the code space of the Iceberg code, i.e. in a state $\ket{\psi}$ which is a stabiliser state of $X^{\otimes 2N},Z^{\otimes 2N}$. The operator $P$ is a logical Pauli operator whose measurement statistics on this state are given by Born rule. For each eigenvalue $\mu\in\{\pm 1\}$ of $P$, we define the projector
\begin{equation}
    \Pi_{\mu}=\frac{I+\mu P}{ 2},
\end{equation}
obtaining a probability of measuring $\mu$ given by
\begin{equation} \label{eq:probdistr}
    \mathrm{Pr}[\mu]=\braket{\psi|\Pi_\mu|\psi}= \frac{1+\mu \langle P\rangle_\psi}2,
\end{equation}
and a post-measurement state given by \begin{equation}\label{eq:postmeasurstate}
    \ket{\psi_{\mu}}\propto \Pi_{\mu} \ket{\psi}.
\end{equation}

The purpose of the protocol is to obtain the correct probability distribution \eqref{eq:probdistr} and the post-measurement state \eqref{eq:postmeasurstate}. The protocol begins by initialising the combined system in the state $\ket{\psi} \ket{0}^{\otimes |E|}$ by resetting the gauge qubits as explained in the main text.

This initial state is not gauge invariant since it does not satisfy the physical constraints $G_v\ket{\varphi}=\ket{\varphi}$. Indeed, the initial state is an eigenstate of the gauge field $\tau_e^Z$, which anticommutes with the vertex operator $G_v$ whenever $v$ is an endpoint of $e$, meaning that $\langle G_v\rangle=0$. The effect of this is that the outcome $g_v$ of measuring each $G_v$ is non-deterministic with equal probability $\pm 1$. However, the random variable $\mu=\prod_v g_v$ has the same probability distribution as measuring $P$ and obtaining $\mu$. We have
\begin{equation}
    \mathrm{Pr}\left[\prod_v g_v=\mu \right]=\sum_{g: \prod_v g_v=\mu} \mathrm{Pr}[g],\quad \mathrm{Pr}[g]=\langle\prod_{v\in V} \frac{I+g_v G_v}{2} \rangle,
\end{equation}
expanding further we find that 
$$\prod_{v\in V} \frac{I+g_v G_v}{2} =\frac{1}{2^{|V|}}\sum_{A\subseteq V}g_A G_A, \quad g_a:= \prod_{v\in A} g_v,\quad G_A=\prod_{v\in A}G_v$$ so that $$\mathrm{Pr}[\mu]=\frac{1}{2^{|V|}}\sum_{A\subseteq V} \left(\sum_{g:\prod_v g_v=\mu}g_A\right)\braket{G_A}$$
Now, considering the quantity in parentheses, we find
\begin{equation}
        S(A):=\sum_{g:\prod_v g_v=\mu} g_A = \sum_{g\in \{\pm 1\}^V} g_A \frac{1+\mu g_V}{2}=\frac{1}{2}\left[\sum_g g_A+\mu \sum_g g_A g_V\right]
\end{equation}
where $g_V=\prod_{v\in V} g_V$. To simplify this expression we note first that 
\begin{equation}
    \sum_{g\in \{\pm 1\}^V} g_A= \prod_{v\in A} \left(\sum_{g_v=\pm 1} g_v\right) \prod_{v\notin A} 2=\begin{cases}
        2^{|V|}&\text{if } A=\emptyset\\
        0  & \text{otherwise}
    \end{cases} 
\end{equation}
Further, using that $g_A g_V=\prod_{v\in A}g_v^2 \prod_{v\notin A}g_v=g_{A^c}$, we obtain 
\begin{equation}
    \sum_{g\in \{\pm 1\}^V} g_A g_V=\begin{cases}
        2^{|V|}&\text{if } A=V\\
        0  & \text{otherwise}
    \end{cases} 
\end{equation}
and therefore 
\begin{equation}
    S(A)=\begin{cases}
        2^{|V|-1}&  A=\emptyset\\
        \mu 2^{|V|-1}&  A=V\\
        0 & \text{otherwise}
    \end{cases}
\end{equation}
Finally, we conclude
\begin{equation}
    \mathrm{Pr}[\mu]=\frac{1}{2}[\braket{G_\emptyset}+\mu \braket{G_V} ]=\frac{1+\mu \braket{P}}{2}
\end{equation}
which is the expected probability distribution.

We conclude that before measurement the state is prepared in a superposition over all possible charge configurations $g$, with each configuration carrying a uniform probability $\mathrm{Pr}[g]=\mathrm{Pr}[\mu(g)]/2^{|V|-1}$ across the $2^{|V|-1}$ configuration sharing a given parity $\mu$. The two parity sectors are weighted overall by the probability $\mathrm{Pr}[\mu]=\frac{1+\mu \braket{P}}{2}$.

In the next step of the protocol, we collapse the state to a particular charge configuration by measuring each $G_v$. The projector in the configuration $g$ is given by \begin{equation}
    \Pi_g =\prod_{v\in V} \frac{I+g_v G_v}{2}=\frac{1}{2^{|V|}}\sum_{A\subseteq V}g_{A}\prod_{v\in A}G_{v}.
\end{equation}
To simplify this expression, we note that
\begin{equation}
\prod_{v\in A}G_{v}=\underbrace{\prod_{v\in A}P_{v}}_{P_{A}} \underbrace{\prod_{v\in A}\prod_{e\ni v}\tau_{e}^X}_{{*}}.
\end{equation}
In particular, an edge contributes $\tau_{e}^X$ if and only if it has endpoints in $A$, so introducing  
\begin{equation}
    \delta(A)=\{ e\in E: |\partial e \cap A|=1 \}\quad (A\subseteq V)
\end{equation} 
we have $$\prod_{v\in A}G_{v}=P_{A}\prod_{e\in \delta(A)}\tau_{e}^X.$$ The post-measurement state is then given by 
\begin{equation}
    \ket{\varphi_g}=\Pi_g \ket{\psi}\ket{0}^{\otimes |E|}=\frac{1}{2^{|V|}}\sum_{A\subseteq V}g_{A}P_{A}\ket{\psi}\otimes \ket{1_{\delta(A)}} .  
\end{equation}
where we have introduced the cut state configuration of the gauge field $\ket{1_{\delta(A)}}:=\prod_{e\in \delta(A)} \tau_e^x \ket{0}^{\otimes |E|}$. Using the fact that $\delta(A)=\delta(A^c)$ so that $P_{A^c}=P_{A}P$ and $g_{A^c}=g_{A}\mu$ we find that $$g_{A}P_{A}\ket{\psi}\ket{1_{\delta(A)}}+g_{A^C}P_{A^C}\ket{\psi}\ket{1_{\delta(A^C)}}=g_{A}P_{A}(I+\mu P)\ket{\psi}\ket{1_{\delta(A)}}      $$so that finally \begin{equation}
    \ket{\tilde{\psi}}=\frac{1}{2^{|V|-1}}\sideset{}{'}\sum_{A\subseteq V }  g_{A}P_{A}\ket{\psi_{\mu}}\ket{1_{\delta(A)}} 
\end{equation} 
where the prime over the sum denotes a summation over $A$ for pairs $(A,A^c)$ and $\ket{\psi_{\mu}}=\frac{I+\mu P}{2}\ket{\psi}$.

The state $\ket{\tilde{\psi}}$ has the charge sector fixed to $\mu$, but the matter remains entangled with the gauge field since each term pairs a Pauli dressing $P_A$ with a field configuration $\ket{1_{\delta(A)}}$ supported on the cut $\delta(A)$. Since the $\ket{1_{\delta(A)}}$ are exactly the flat field configurations, this is a coherent superposition over all gauge-equivalent representatives of the same fixed charge---a redundancy the following step removes by measuring the gauge field directly.

Measuring the gauge field in the $\tau^Z$ basis --- the third step of the protocol --- acts as a gauge-fixing: it collapses the superposition over gauge-equivalent cut representatives onto a single one, disentangling matter from gauge, and the state collapses to one satisfying $m=\delta(A)$, where $m_e=\frac{1-w_e}{2}$ with $w_e=\pm1$ denoting the outcome of the measurement. We can discard the gauge qubits and we are left with $P_A \ket{\psi_\mu}$ up to phase. 

The final step is to determine $A$ given $m$. The solution is not unique---if $A$ is a solution then so is $A^c$---but we have in either case $P_{A^{c}}= P_A P$, so $P_{A^c} P_A \ket{\psi_\mu}\propto \ket{\psi_\mu}$. We obtain the following proposition
\begin{proposition}
     A solution to $m = \delta(A)$  exists for a connected graph if and only if $\sum_{e \in c} m_e=0 \pmod{2}$ for every cycle $c$ of the graph.
\end{proposition}
\begin{proof}
Let us first give some definitions. The graph $G=(V,E)$ may be encoded algebraically by two $\mathbb{F}_2$-vector spaces, $C_0:= \mathbb{F}_2^{|V|}$ and $C_1:= \mathbb{F}_2^{|E|}$, with canonical bases given by the vertex set $V$ and edge set $E$ respectively: a basis vector of $C_0$ corresponds to a single vertex $v$, and a basis vector of $C_1$ corresponds to a single edge $e$. A general element of $C_0$ is then a vertex subset $A\subseteq V$, identified with its indicator vector $1_A\in C_0$; likewise a general element of $C_1$ is an edge subset, identified with its indicator vector.

We can then define a linear map between these two vector space 
\begin{equation}
\begin{split}
    &\partial: C_1  \to C_0\\
    &\partial(e)=u+v\pmod{2}\quad \text{for } e=(u,v)
\end{split}
\end{equation}
extended linearly to all of $C_1$. This map is a boundary map, and with respect to the canonical bases it is represented by the incidence matrix of the graph. In addition, we define the two bilinear maps \begin{equation}
    \begin{split}
        &\braket{\cdot,\cdot}_{C_0}: C_0\to \mathbb{F}_2\\
        &\braket{x,y}_{C_{0}}=\sum_{v\in V}x_{v}y_{v}\pmod{2} 
    \end{split}
\end{equation}
and \begin{equation}
    \begin{split}
            &\braket{\cdot,\cdot}_{C_1}: C_1\to \mathbb{F}_2\\
        &\braket{u,v}_{C_{1}}=\sum_{e\in E} u_{e}v_{e} \pmod{2}
    \end{split}
\end{equation}
which allow us to define the coboundary map $\delta:C_0\to C_1$ as the adjoint of $\partial$, i.e. the map satisfying\begin{equation}
    \langle \partial u,x \rangle_{C_{0}}=\langle u,\delta x \rangle_{C_{1}}\quad \forall u\in C_{1},x\in C_{0}  
\end{equation}
Concretely, given the vertex set $A$ which we represent by the vector $1_{A}$, we have $\delta(1_{A})_{e}=1$ if and only if $e$ has exactly one endpoint in $A$. Finally, we introduce the cycle space $Z_1:= \ker \partial \subseteq C_1$, since for each cycle $c$ we have $\partial c =0$ and the cut space $B^{1}=\mathrm{im} \delta \subseteq C_1$.

Now let us suppose that $m=\delta(A)$ for some $A\subseteq V$, then for any cycle $c$ of the graph we have that $\braket{m,c}_{C_1}=\braket{\delta(A),c}_{C_1}=\braket{A,\partial c}_{C_0} =0$ where we use $m=\delta(A)$ in the first equality and $\partial c=0$ in the last. This means that $\sum_{e\in c} m_e = 0$ and so $\prod_{e\in c} (-1)^{m_e}=1$ which is what we expect from consistency with $B_c=+1$ for each cycle.

Conversely, let us consider a spanning tree $T$ of a connected graph, i.e. a subset of edges such that $(V,T)$ is an acyclic graph including every vertex with a specified root vertex $r\in V$. For each $v\in V$ we define the unique tree path $\Gamma_v\subseteq T$ that spans from $r$ to $v$ and the function \begin{equation}
        \varphi_v = \braket{m,\Gamma_v}_{C_1}=\sum_{e\in \Gamma_v}m_e \pmod{2}.
    \end{equation}
We then construct the sets \begin{equation}
        A=\{v\in V: \varphi_v = 1\}
    \end{equation} 
    which solve the equation $m=\delta(A)$. In particular, for every edge $e\in E$ we have $m_{e}=1$ if and only if  $e\in \delta(A)$.

Now let us consider a tree edge $e=(u,v)\in T$. We have that $\Pi_v = \Pi_u \cup \{e\}$ which gives \begin{equation}
        \varphi_v= \braket{m,\Pi_v} = \braket{m,\Pi_u} +m_e = \varphi_u \oplus m_e.
    \end{equation} 
    
Equivalently, $m_e = \varphi_u \oplus \varphi_v$ which is $1$ exactly when $\varphi_{u}\neq \varphi_{v}$, i.e. exactly when one of $u,v$ is in $A$ and the other is in the complement. We conclude $m_{e}=[e\in \delta(A)]$ for every tree edge $e$.
    
Now, let $e$ be a non tree edge, $e=(u,v)\in E\setminus T$. Since $T\cup \{ e \}$ contains a unique cycle (adding any edge to a spanning tree creates exactly one cycle), which we denote $c := \Pi_u \,\triangle\, \{e\} \,\triangle\, \Pi_v$ where $\overline{\Pi}_{v}$ is the tree path $\Pi_{v}$ traversed in the reverse direction. Since $\langle m,c \rangle=0$ we have $$0=\langle m,c \rangle=\langle m,\Pi_{u} \rangle \oplus \langle m,\Pi_{v} \rangle\oplus m_{e}=\varphi_{u}\oplus \varphi_{v}\oplus m_{e}   $$ and we conclude that $m_{e}=[e\in \delta(A)]$.
\end{proof}

The above proposition demonstrates how $A$ may be determined from the measurement of the gauge field. After determining $A$, we apply $P_A$ to recover the post measurment state $\ket{\psi_\mu}$.

\bibliographystyle{unsrtnat}
\bibliography{refs}

@article{self2024protecting,
  title={Protecting expressive circuits with a quantum error detection code},
  author={Self, Chris N and Benedetti, Marcello and Amaro, David},
  journal={Nature Physics},
  volume={20},
  number={2},
  pages={219--224},
  year={2024},
  publisher={Nature Publishing Group UK London}
}

@article{chao2018fault,
  title={Fault-tolerant quantum computation with few qubits},
  author={Chao, Rui and Reichardt, Ben W},
  journal={npj Quantum Information},
  volume={4},
  number={1},
  pages={42},
  year={2018},
  publisher={Nature Publishing Group UK London}
}

@article{horsman2012surface,
  title={Surface code quantum computing by lattice surgery},
  author={Horsman, Dominic and Fowler, Austin G and Devitt, Simon and Meter, Rodney Van},
  journal={New Journal of Physics},
  volume={14},
  number={12},
  pages={123011},
  year={2012},
  publisher={IOP Publishing}
}

@article{bravyi2016trading,
  title={Trading classical and quantum computational resources},
  author={Bravyi, Sergey and Smith, Graeme and Smolin, John A},
  journal={Physical Review X},
  volume={6},
  number={2},
  pages={021043},
  year={2016},
  publisher={APS}
}

@article{williamson2026low,
  title={Low-overhead fault-tolerant quantum computation by gauging logical operators},
  author={Williamson, Dominic J and Yoder, Theodore J},
  journal={Nature Physics},
  pages={1--6},
  year={2026},
  publisher={Nature Publishing Group UK London}
}

@article{chamberland2018flag,
  title={Flag fault-tolerant error correction with arbitrary distance codes},
  author={Chamberland, Christopher and Beverland, Michael E},
  journal={Quantum},
  volume={2},
  pages={53},
  year={2018},
  publisher={Verein zur F{\"o}rderung des Open Access Publizierens in den Quantenwissenschaften}
}

@article{vuillot2019code,
  title={Code deformation and lattice surgery are gauge fixing},
  author={Vuillot, Christophe and Lao, Lingling and Criger, Ben and Garc{\'\i}a Almud{\'e}ver, Carmen and Bertels, Koen and Terhal, Barbara M},
  journal={New Journal of Physics},
  volume={21},
  number={3},
  pages={033028},
  year={2019},
  publisher={IOP Publishing}
}

@article{gidney2021stim,
  title={Stim: a fast stabilizer circuit simulator},
  author={Gidney, Craig},
  journal={Quantum},
  volume={5},
  pages={497},
  year={2021},
  publisher={Verein zur F{\"o}rderung des Open Access Publizierens in den Quantenwissenschaften}
}

@article{bravyi2005universal,
  title={Universal quantum computation with ideal Clifford gates and noisy ancillas},
  author={Bravyi, Sergey and Kitaev, Alexei},
  journal={Physical Review A—Atomic, Molecular, and Optical Physics},
  volume={71},
  number={2},
  pages={022316},
  year={2005},
  publisher={APS}
}

@article{sayginel2025fault,
  title={Fault-tolerant logical clifford gates from code automorphisms},
  author={Sayginel, Hasan and Koutsioumpas, Stergios and Webster, Mark and Rajput, Abhishek and Browne, Dan E},
  journal={PRX Quantum},
  volume={6},
  number={3},
  pages={030343},
  year={2025},
  publisher={APS}
}

@article{berthusen2025automorphism,
  title={Automorphism gadgets in homological product codes},
  author={Berthusen, Noah and Gullans, Michael J and Hong, Yifan and Mudassar, Maryam and Tan, Shi Jie Samuel},
  journal={arXiv preprint arXiv:2508.04794},
  year={2025}
}

@article{bluvstein2022quantum,
  title={A quantum processor based on coherent transport of entangled atom arrays},
  author={Bluvstein, Dolev and Levine, Harry and Semeghini, Giulia and Wang, Tout T and Ebadi, Sepehr and Kalinowski, Marcin and Keesling, Alexander and Maskara, Nishad and Pichler, Hannes and Greiner, Markus and others},
  journal={Nature},
  volume={604},
  number={7906},
  pages={451--456},
  year={2022},
  publisher={Nature Publishing Group UK London}
}

@article{bluvstein2024logical,
  title={Logical quantum processor based on reconfigurable atom arrays},
  author={Bluvstein, Dolev and Evered, Simon J and Geim, Alexandra A and Li, Sophie H and Zhou, Hengyun and Manovitz, Tom and Ebadi, Sepehr and Cain, Madelyn and Kalinowski, Marcin and Hangleiter, Dominik and others},
  journal={Nature},
  volume={626},
  number={7997},
  pages={58--65},
  year={2024},
  publisher={Nature Publishing Group UK London}
}

@article{cohen2022low,
  title={Low-overhead fault-tolerant quantum computing using long-range connectivity},
  author={Cohen, Lawrence Z and Kim, Isaac H and Bartlett, Stephen D and Brown, Benjamin J},
  journal={Science Advances},
  volume={8},
  number={20},
  pages={eabn1717},
  year={2022},
  publisher={American Association for the Advancement of Science}
}

@article{cross2024improved,
  title={Improved QLDPC surgery: Logical measurements and bridging codes},
  author={Cross, Andrew W and He, Zhiyang and Rall, Patrick J and Yoder, Theodore J},
  journal={arXiv preprint arXiv:2407.18393},
  year={2024}
}

@book{gottesman1997stabilizer,
  title={Stabilizer codes and quantum error correction},
  author={Gottesman, Daniel},
  year={1997},
  publisher={California Institute of Technology}
}

@article{poulin2005stabilizer,
  title={Stabilizer formalism for operator quantum error correction},
  author={Poulin, David},
  journal={Physical review letters},
  volume={95},
  number={23},
  pages={230504},
  year={2005},
  publisher={APS}
}

@article{bombin2009quantum,
  title={Quantum measurements and gates by code deformation},
  author={Bomb{\'\i}n, H{\'e}ctor and Martin-Delgado, Miguel Angel},
  journal={Journal of Physics A: Mathematical and Theoretical},
  volume={42},
  number={9},
  pages={095302},
  year={2009}
}

@article{cowtan2024css,
  title={CSS code surgery as a universal construction},
  author={Cowtan, Alexander and Burton, Simon},
  journal={Quantum},
  volume={8},
  pages={1344},
  year={2024},
  publisher={Verein zur F{\"o}rderung des Open Access Publizierens in den Quantenwissenschaften}
}

@article{gottesman1998theory,
  title={Theory of fault-tolerant quantum computation},
  author={Gottesman, Daniel},
  journal={Physical Review A},
  volume={57},
  number={1},
  pages={127},
  year={1998},
  publisher={APS}
}

@article{swaroop2026universal,
  title={Universal adapters between quantum low-density parity check codes},
  author={Swaroop, Esha and Jochym-O'Connor, Tomas and Yoder, Theodore J},
  journal={PRX Quantum},
  volume={7},
  number={1},
  pages={010324},
  year={2026},
  publisher={APS}
}

@article{reichardt2024logical,
  title={Logical computation demonstrated with a neutral atom quantum processor},
  author={Reichardt, Ben W and Paetznick, Adam and Aasen, David and Basov, Ivan and Bello-Rivas, Juan M and Bonderson, Parsa and Chao, Rui and van Dam, Wim and Hastings, Matthew B and Paz, Andres and others},
  journal={arXiv preprint arXiv:2411.11822},
  volume={10},
  year={2024}
}

@article{zhang2026logical,
  title={Logical qubits with erasure conversion using metastable neutral atoms},
  author={Zhang, Bichen and Liu, Genyue and Bornet, Guillaume and Horvath, Sebastian P and Peng, Pai and Ma, Shuo and Huang, Shilin and Puri, Shruti and Thompson, Jeff D},
  journal={Nature Physics},
  volume={22},
  number={6},
  pages={910--916},
  year={2026},
  publisher={Nature Publishing Group}
}

@article{lib2026velocity,
  title={Velocity-enabled quantum computing with neutral atoms},
  author={Lib, Ohad and Timme, Hendrik and Ammenwerth, Maximilian and Gyger, Flavien and Tao, Renhao and Sun, Shijia and Bloch, Immanuel and Zeiher, Johannes},
  journal={arXiv preprint arXiv:2603.15561},
  year={2026}
}

@article{mathiot2026benchmarking,
  title={Benchmarking a machine-learning differential equations solver on a neutral-atom logical processor},
  author={Mathiot, Pauline and Garnaoui, Elio and Leriche, Axel-Ugo and Philip, Evan and Albrecht, Boris and Briosne-Fr{\'e}javille, Cl{\'e}mence and Cardarelli, Lorenzo and Cornillot, Antoine and Cournez, Gwennol{\'e} and Couturier, Luc and others},
  journal={arXiv preprint arXiv:2605.21276},
  year={2026}
}

@article{webster2025explicit,
  title={Explicit construction of low-overhead gadgets for gates on quantum LDPC codes},
  author={Webster, Paul and Smith, Samuel C and Cohen, Lawrence Z},
  journal={arXiv preprint arXiv:2511.15989},
  year={2025}
}

@article{kogut1975hamiltonian,
  title={Hamiltonian formulation of Wilson's lattice gauge theories},
  author={Kogut, John and Susskind, Leonard},
  journal={Physical Review D},
  volume={11},
  number={2},
  pages={395},
  year={1975},
  publisher={APS}
}

@article{kogut1979introduction,
  title={An introduction to lattice gauge theory and spin systems},
  author={Kogut, John B},
  journal={Reviews of Modern Physics},
  volume={51},
  number={4},
  pages={659},
  year={1979},
  publisher={APS}
}

\end{document}